\documentclass[english,a4paper,10pt]{article}%{ceurart}
\usepackage[utf8]{inputenc}

\usepackage[T1]{fontenc}
\usepackage[a4paper, total={6in, 8in}]{geometry}

\usepackage{amsmath}
\usepackage[english]{babel}

\usepackage{framed}
\usepackage{amssymb}
\usepackage{todonotes}
\usepackage{hyperref}
\usepackage{cleveref}
\crefname{algocf}{Algorithm}{Algorithms}
\usepackage[vlined,ruled,noend,linesnumbered]{algorithm2e}
\SetArgSty{text}
\DontPrintSemicolon
\usepackage{bussproofs}
\EnableBpAbbreviations
\usepackage{mathtools}
\usepackage{scalerel}
\usepackage{tikz}
\usetikzlibrary{shapes.geometric, patterns}
\usetikzlibrary{arrows.meta}
\usetikzlibrary{calc}
\usetikzlibrary{shapes.misc, positioning}
\usepackage{csquotes}
\usepackage{booktabs}
\usepackage{thmtools}
\usepackage{thm-restate}
\usepackage{enumitem}
\usepackage{pdfpages}
\usepackage{nameref}
\usepackage{refcount}

\usepackage{stmaryrd}

\newif\iftr
\trtrue
\usepackage{amsthm}
\usepackage[strings]{underscore}
 \newtheorem{theorem}{Theorem}

 \newtheorem{example}{Example}

\newcommand{\patrick}[2][inline]{\ttodo{#1}{Patrick}{green}{#2}}

\NewCommandCopy{\oldtodo}{\todo}
\newcommand{\ttodo}[4]{\ifthenelse{\equal{#1}{inline}}{\oldtodo[inline,color = #3]{#2: #4}}{\oldtodo[color=#3]{#2: #4}}}
\renewcommand{\todo}[2][inline]{\ttodo{#1}{TODO}{orange}{#2}}

\newcommand{\hide}[1]{}

\newcommand{\wrt}{w.r.t.\ }
\newcommand{\st}{s.t.\ }
\newcommand{\ie}{i.e.\ }
\newcommand{\eg}{e.g.\ }
\newcommand{\etc}{etc.\xspace}
\newcommand{\cf}{cf.\ }
\newcommand{\vs}{vs.\ }
\newcommand{\aka}{a.k.a.\ }

\newlength{\myl}
\newcommand{\longsquigarrow}[1]{
    \settowidth{\myl}{$~_{#1}$}
    \raisebox{-0.01cm}{\xymatrix@C=\myl{
            {}\ar@{~>}[r]^{~_{#1}}&{}
        }
    }
}

\newcommand{\PTime}{\text{\upshape{\textsc{P}}}\xspace}

\newcommand{\NP}{\text{\upshape{\textsc{NP}}}\xspace}

\newcommand{\ExpTime}{\text{\upshape{\textsc{ExpTime}}}\xspace}

\newcommand{\NExpTime}{\text{\upshape{\textsc{NExpTime}}}\xspace}

\newcommand{\Cmc}{\ensuremath{\mathcal{C}}\xspace}
\newcommand{\Dmc}{\ensuremath{\mathcal{D}}\xspace}
\newcommand{\Imc}{\ensuremath{\mathcal{I}}\xspace}

\newcommand{\Lmc}{\ensuremath{\mathcal{L}}\xspace}

\newcommand{\Omc}{\ensuremath{\mathcal{O}}\xspace} 
 
\newcommand{\Tmc}{\ensuremath{\mathcal{T}}\xspace}

\newcommand{\Nbb}{\ensuremath{\mathbb{N}}\xspace} 
\newcommand{\Qbb}{\ensuremath{\mathbb{Q}}\xspace}

\newcommand{\ALC}{\ensuremath{\mathcal{ALC}}\xspace}

\newcommand{\EL}{\ensuremath{\mathcal{E}\hspace{-0.1em}\mathcal{L}}\xspace}
\newcommand{\ELbot}{\ensuremath{\EL_\bot}\xspace}

\newcommand{\tup}[1]{\langle #1 \rangle}

\newcommand{\NC}{\ensuremath{\textsf{N}_\textsf{C}}\xspace}
\newcommand{\NR}{\ensuremath{\textsf{N}_\textsf{R}}\xspace}

\newcommand{\NV}{\ensuremath{\textsf{N}_\textsf{V}}\xspace}
\newcommand{\NPred}{\ensuremath{\textsf{N}_\textsf{P}}\xspace}

\newcommand{\ex}[1]{\ensuremath{\textsf{\upshape{#1}}}\xspace}
 
\newcommand{\p}{\ensuremath{\mathcal{P}}\xspace}

\newcommand{\R}{\ensuremath{\mathfrak{D}}\xspace}

\newcommand{\true}{\ensuremath{\mathsf{true}}\xspace}
\newcommand{\false}{\ensuremath{\mathsf{false}}\xspace}
\newcommand{\Lethe}{\textsc{Lethe}\xspace}
\newcommand{\Elk}{\textsc{Elk}\xspace}

\newcommand{\DQlin}{\ensuremath{\mathcal{D}_{\mathbb{Q},\textit{lin}}}\xspace}
\newcommand{\DQgr}{\ensuremath{\mathcal{D}_{\mathbb{Q},\textit{diff}}}\xspace}

\newcommand{\ELbotD}[1][\Dmc]{\ensuremath{\ELbot[#1]}\xspace}
\newcommand{\ALCD}[1][\Dmc]{\ensuremath{\ALC[#1]}\xspace}
\newcommand{\ar}{\textsf{ar}\xspace}

\newcommand{\class}{\ensuremath{\mathsf{CL}}\xspace}
\newcommand{\sub}{\ensuremath{\mathsf{sub}}\xspace}
\newcommand{\Nbf}{\ensuremath{\mathcal{N}}\xspace}
\newcommand{\Dbf}{\ensuremath{\mathbf{D}}\xspace}

\newcommand{\Rnotequal}{\ensuremath{\mathsf{R_{\neq}}}\xspace}
\newcommand{\Rsmaller}{\ensuremath{\mathsf{R_<}}\xspace}
\newcommand{\Rnotequalplus}{\ensuremath{\mathsf{R_{\neq}^+}}\xspace}
\newcommand{\Rzero}{\ensuremath{\mathsf{R_0}}\xspace}
\newcommand{\Rdiff}{\ensuremath{\mathsf{R_-}}\xspace}
\newcommand{\Rmirror}{\ensuremath{\mathsf{R_\leftrightarrow}}\xspace}
\newcommand{\Rtrans}{\ensuremath{\mathsf{R_+}}\xspace}
\newcommand{\Requal}{\ensuremath{\mathsf{R_=}}\xspace}
\newcommand{\Rgreater}{\ensuremath{\mathsf{R_>}}\xspace}
\newcommand{\Rbot}{\ensuremath{\mathsf{R_\bot}}\xspace}
\newcommand{\Rconstantbound}{\ensuremath{\mathsf{R_>^+}}\xspace}
\newcommand{\Rweaken}{\ensuremath{\mathsf{R_>^-}}\xspace}

\newcommand{\saturate}{\ensuremath{\mathsf{saturate}}\xspace}

\begin{document}
\iftr{
    \title{Combining Proofs for Description Logic and Concrete Domain 
    Reasoning (Technical Report)}
}
\else{
    \title{Combining Proofs for Description Logic and Concrete Domain Reasoning}
}
\fi
%
%\titlerunning{Proofs for DLs with Concrete Domains}
%
%\author{Christian Alrabbaa\inst{1}
%%\orcidID{0000-0002-2925-1765}
%\and
%Franz Baader\inst{1}
%\and
%Stefan Borgwardt\inst{1}
%%\orcidID{0000-0003-0924-8478}
%\and
%Patrick Koopmann\inst{2}
%%\orcidID{0000-0001-5999-2583}
%\and
%Alisa Kovtunova\inst{1}%\orcidID{0000-0001-9936-0943}
%}
%
%\authorrunning{C.\ Alrabbaa et al.}
%
%\institute{Institute of Theoretical Computer Science, TU Dresden, Germany\\
%    \email{firstname.lastname@tu-dresden.de} \and 
%Department of Computer Science, Vrije Universiteit Amsterdam, Netherlands\\
 %   \email{p.k.koopmann@vu.nl}
%}
%
\author{Christian Alrabbaa \and Franz Baader \and Stefan Borgwardt \and Patrick Koopmann \and Alisa Kovtunova}%[%

\date{}
\maketitle

\begin{abstract}
Logic-based approaches to AI have the advantage that their behavior can in principle be explained %to a user 
with the help of proofs of the computed consequences.
% in an appropriate calculus. To benefit from this in practice, considerable work beyond the implementation of a reasoning system is needed to be 
%able to compute proofs that are appropriate for explanation purposes. 
For ontologies based on Description Logic (DL), we have put this advantage into practice
%In recent work, we have investigated 
by showing how proofs for consequences derived by DL reasoners can be computed and displayed in a user-friendly way. 
 %
%However, this work was restricted to DLs without concrete domains, and is thus not sufficient for applications where concrete objects (such as numbers) and predefined predicates on these objects (such as numerical comparisons) are needed to define the relevant concepts. 
 %
However, these methods are insufficient in applications where also numerical reasoning is relevant.
The present paper considers proofs for DLs extended with concrete domains (CDs) based on the rational numbers, which leave reasoning tractable if integrated into the lightweight DL \ELbot.
%that allow the integration of reasoning over the rational numbers,
%---one supporting linear constraints,
%and one that is related to difference logic. Both leave reasoning tractable if integrated into the lightweight DL \ELbot.
 %
%Here we consider two concrete domains (CDs) based on the rational numbers, which leave reasoning tractable when integrated into the light-weight DL \ELbot. 
 %
Since no implemented DL reasoner supports these CDs, we first develop reasoning procedures for them, and show how they can be combined with reasoning approaches for pure DLs, both for \ELbot and the more expressive DL \ALC. These procedures are designed such that it is easy to extract proofs from them. We show how the extracted CD proofs can be combined with proofs on the DL side
into integrated proofs that explain both the DL and the CD reasoning. % in a uniform proof format.
 %
%We have implemented our reasoning and proof extraction approaches for DLs with concrete domains and have evaluated them on several self-created benchmarks.
\end{abstract}

\section{Introduction}

Description Logics (DLs) \cite{DBLP:books/daglib/0041477} are a well-investigated family of logic-based knowledge representation languages,
which are frequently used to formalize ontologies for various application domains.
%such as the Semantic Web \cite{HoPH03} or biology and medicine~\cite{HoSG15}. 
As the sizes of DL-based ontologies grow, tools that support improving the quality of such ontologies become more important.
DL reasoners\footnote{%
See \href{http://owl.cs.manchester.ac.uk/tools/list-of-reasoners/}{http://owl.cs.manchester.ac.uk/tools/list-of-reasoners/}}
%\cite{Horr98,HaMo01b,SiPa04}
can be used to detect inconsistencies and to infer other implicit consequences, such as subsumption relationships.
However, for developers or users of DL-based ontologies, it is often hard to understand why a consequence computed by the reasoner
actually follows from the given, possibly very large ontology. In principle, such a consequence can be explained by producing a proof for it,
%i.e., by showing 
which shows how the consequence can be derived from the axioms in the ontology by applying certain easy-to-understand inference rules.
 %
%Until recently, work on explaining DL entailment was focused on computing so-called justifications, i.e., minimal subsets of the ontology from which the 
%consequence in question follows~\cite{ScCo03,DBLP:conf/ki/BaaderPS07,Horr-11}. With few exceptions~\cite{DBLP:conf/semweb/HorridgePS10,DBLP:conf/dlog/KazakovKS17}, 
%figuring out how the consequence can be derived from the justification was left to the user.
 %
In recent work, we have investigated how proofs for consequences derived by DL reasoners can be computed~\cite{DBLP:conf/lpar/AlrabbaaBBKK20,DBLP:conf/cade/AlrabbaaBBKK21}
and 
displayed~\cite{https://doi.org/10.1111/cgf.14730} %~\cite{DBLP:conf/cade/AlrabbaaBBDKM22}
 in a user-friendly 
way~\cite{DBLP:conf/ruleml/AlrabbaaBHKKRW22}.
However, like previous work~\cite{DBLP:conf/semweb/HorridgePS10,DBLP:conf/dlog/KazakovKS17}, this was restricted to DLs without concrete domains.

Concrete domains~\cite{DBLP:conf/ijcai/BaaderH91,DBLP:conf/aiml/Lutz02} (CDs) have been introduced %in Description Logic
 to enable reference to concrete objects (such as numbers) 
and predefined predicates on these objects (such as numerical comparisons) when defining concepts. For example, assume that we measure the systolic and the diastolic
blood pressure of patients. Then we can describe patients with a pulse pressure of 25\,mmHg as %extremely low pulse pressure as
$
\ex{Patient}\sqcap[\ex{sys}-\ex{dia} = 25],
$
where $\ex{sys}$ and $\ex{dia}$ are \emph{features} that are interpreted as partial functions that return the systolic and the diastolic blood pressure
of a patient, respectively, as rational
numbers (if available). We can then state that such patients need attention using the general concept inclusion (GCI) 
$$
\ex{Patient}\sqcap[\ex{sys}-\ex{dia} = 25] \sqsubseteq \ex{NeedAttention}.
$$
In the presence of GCIs, integrating a CD into a DL may cause undecidability~\cite{DBLP:journals/tocl/Lutz04,DBLP:journals/jar/BaaderR22} even if solvability of the 
constraint systems that can be formulated in the CD (in our example, sets of constraints of the form $x-y = q$ for $q\in\Qbb$) %~\cite{DBLP:series/txtcs/KroeningS16}) 
is decidable. One way to overcome this problem is to disallow role paths~\cite{DBLP:conf/cade/HaarslevMW01,DBLP:conf/dlog/PanH02,DBLP:conf/ijcai/BaaderBL05} in concrete domain restrictions, 
which means that these restrictions can only constrain feature values of single individuals, as in our example. Comparing feature values of different individuals, such as the age of a woman 
with that of her children, is then no longer possible.
 
For tractable (i.e., polynomially decidable) DLs like \ELbot, preserving decidability is not sufficient: one wants to preserve tractability. As shown in~\cite{DBLP:conf/ijcai/BaaderBL05},
this is the case if one integrates a so-called p-admissible concrete domain into \ELbot.
The only numerical p-admissible concrete domain exhibited in~\cite{DBLP:conf/ijcai/BaaderBL05} is the CD
$\DQgr$, which supports constraints of the form $x=q$, $x>q$, and $x+q=y$ (for constants $q\in\Qbb$). % and thus can express the difference constraints used in our example.
Recently, additional p-admissible concrete domains have been introduced in~\cite{DBLP:journals/jar/BaaderR22}, such as \DQlin, whose constraints are given by linear equations 
$\sum_{i=1}^na_ix_i=b$. In the present paper, we will concentrate on these two p-admissible CDs, though the developed ideas and techniques can also be used for other CDs.
The constraint used in  our example can be expressed in both \DQgr and \DQlin. % of these CDs.
Unfortunately, no implemented DL reasoner supports these two CDs. In particular, the highly efficient \ELbot reasoner \Elk~\cite{DBLP:journals/jar/KazakovKS14} does not support any 
concrete domain. Instead of modifying \Elk or implementing our own reasoner for \ELbot with concrete domains, we develop here an iterative algorithm that interleaves
\Elk reasoning with concrete domain reasoning. For the CD reasoning, we could in principle employ existing algorithms and implementations, like Gaussian elimination or 
the simplex method~\cite{Turner1995GaussEW,DBLP:conf/cav/DutertreM06} for \DQlin, and SMT systems that can deal with 
difference logic~\cite{DBLP:series/txtcs/KroeningS16,DBLP:conf/sat/ArmandoCGM04a}, such as Z3,\footnote{%
\href{https://theory.stanford.edu/~nikolaj/programmingz3.html}{https://theory.stanford.edu/\~{}nikolaj/programmingz3.html}}
for $\DQgr$. However, since our main purpose is to generate proofs, we develop our own reasoning procedures for $\DQgr$ and \DQlin, which may not be as efficient
as existing ones, but can easily be adapted such that they produce proofs.

Proofs for reasoning results in \ELbot with a p-admissible CD can in principle be represented using the calculus introduced in~\cite{DBLP:conf/ijcai/BaaderBL05} or an appropriate
extension of the calculus employed by \Elk. However, in these calculi, the result of CD reasoning (i.e., that a set of constraints is unsatisfiable or entails another constraint)
is used as an applicability condition for certain rules, % application, 
but the CD reasoning leading to the satisfaction of the conditions is not explained. Instead of augmenting such a proof with separate
proofs on the CD side that show why the applicability conditions are satisfied, our goal is to produce a single proof that explains both the \ELbot and the CD reasoning in
a uniform proof format. 
%In addition to tackling this problem in the setting of \ELbot, we also consider the integration of the CDs $\DQgr$ and $\DQlin$ into the more expressive DL \ALC. 

We also consider the integration of the CDs $\DQgr$ and $\DQlin$ into the more expressive DL \ALC. To this purpose, we develop a new calculus for subsumption w.r.t.\ \ALC 
ontologies, which is inspired by the one in~\cite{FORGETTING_ALC}, but has a better worst-case complexity,
and then show how it can be extended to deal with concrete domain restrictions.
We have implemented our reasoning and proof extraction approaches for DLs with concrete domains and have evaluated them on several self-created benchmarks
designed specifically to challenge the CD reasoning and proof generation 
capabilities.
\iftr{More details about the experiments can be found in~\cite{zenodo}.}\else{Proofs for all results and more details about the experiments can be found in~\cite{ourarxive,zenodo}. }\fi
 %\url{https://lat.inf.tu-dresden.de/~alrabbaa/rulemlrr23/cdProofs.html}.

%\stefan{comment on experiments!}
% In the next section, we \ldots

%\todo{Franz}
%
%\stefan{\url{https://doi.org/10.1093/logcom/exac085}: size of counterinterpretations for modal logics}
%\stefan{check SMT literature (e.g. they use justifications for communicating between theory solver and SAT solver)}

%!TeX root=main.tex

\section{Description Logics with Concrete Domains}

We recall the DLs \ELbot and \ALC~\cite{DBLP:books/daglib/0041477}, and then discuss their extensions \ELbotD and \ALCD with a concrete domain~\Dmc~\cite{DBLP:conf/ijcai/BaaderH91,DBLP:conf/ijcai/BaaderBL05}. Following~\cite{DBLP:journals/jar/BaaderR22}, we use square brackets to indicate that no role paths are allowed. We also introduce the two p-admissible concrete domains \DQgr and \DQlin~\cite{DBLP:conf/ijcai/BaaderBL05,DBLP:journals/jar/BaaderR22}.

\subsection{Description Logics}
%\subsubsection*{Description Logics.}
%
% \paragraph{DL Syntax.}
%
Starting with disjoint, countably infinite sets of \emph{concept} and \emph{role names} \NC and  \NR,
%Let \NC and  \NR be disjoint, countably infinite sets of \emph{concept} and \emph{role names}. %and \emph{individual names}, respectively.
%
\emph{\ELbot concepts} are defined by the grammar $C,D::=\top\mid\bot\mid A\mid %\lnot C\mid 
C\sqcap D\mid \exists r.C$, where $A\in\NC$ and $r\in\NR$. In \ALC, we additionally have negation $\lnot C$ as  concept constructor.
As usual, we then define $C\sqcup D:=\lnot(\lnot C\sqcap\lnot D)$ and $\forall r.C:=\lnot\exists r.\lnot C$.
An \emph{\ALC (\ELbot) TBox} (\aka \emph{ontology})~\Omc is a finite set of \emph{general concept inclusions (GCIs}, \aka \emph{axioms)} $C\sqsubseteq D$ for \ALC (\ELbot) concepts~$C$ and~$D$.
% %
% An \emph{\ALC ABox} is a finite set of \emph{assertions} $a:C$ or $(a,b):r$, where $a,b\in\NI$, $r\in\NR$, and $C$ is a concept.
% %
% An \emph{\ALC ontology} $\Omc=\Amc\cup\Tmc$ consists of an ABox~\Amc and a TBox~\Tmc.
% %
% An \emph{axiom} is either a GCI or an assertion.
%
We denote by $\sub(\Omc)$ the set of subconcepts of all concepts appearing in~\Omc.
%
% The DL \ELbot restricts \ALC by disallowing the concept constructor~$\lnot$ (and therefore also~$\sqcup$ and~$\forall$).
% \todo{Remove ABoxes? They are good for examples, but then we also need to describe their influence on the algorithms.}
% \patrick{I would leave them out to keep things simple. We could mention in the conclusion that are methods can easily be extended to ABoxes, for instance because ELK already supports them. The term \emph{ABox clause} I later use should not be a problem for leaving ABoxes out here, either.}

% \paragraph{DL Semantics.}
%
An \emph{interpretation} is a pair $\Imc=(\Delta^\Imc,\cdot^\Imc)$, where the \emph{domain} $\Delta^\Imc$ is a non-empty set, and the \emph{interpretation function} $\cdot^\Imc$ assigns to every concept name $A\in\NC$ a set $A^\Imc\subseteq\Delta^\Imc$ and to every role name $r\in\NR$ a binary relation $r^\Imc\subseteq\Delta^\Imc\times\Delta^\Imc$. % and to every individual name $a\in\NI$ an element $a^\Imc\in\Delta^\Imc$.
This function is extended to complex concepts by defining $\top^\Imc:=\Delta^\Imc$, $\bot^\Imc:=\emptyset$, 
$(\exists r.C)^\Imc:=\{d\in\Delta^\Imc\mid\exists e\in\Delta^\Imc.\,(d,e)\in r^\Imc\land e\in C^\Imc\}$, 
$(\lnot C)^\Imc=\Delta^\Imc\setminus C^\Imc$, and
$(C\sqcap D)^\Imc:=C^\Imc\cap D^\Imc$. 
%, and $(\exists r.C)^\Imc:=\{d\in\Delta^\Imc\mid\exists e\in\Delta^\Imc.\,(d,e)\in r^\Imc\land e\in C^\Imc\}$.
%
%Following this, $\top$ and $\bot$ can be treated as the empty conjunction and empty disjunction, respectively.
%
The interpretation~\Imc is a \emph{model} of $C\sqsubseteq D$ if $C^\Imc\subseteq D^\Imc$ (written $\Imc\models C\sqsubseteq D$),
%; of $a:C$ if $a^\Imc\in C^\Imc$; of $(a,b):r$ if $(a^\Imc,b^\Imc)\in r^\Imc$;
and it is a model of an ontology~\Omc ($\Imc\models\Omc$) if it is a model of all axioms in~\Omc.
An ontology~\Omc is \emph{consistent} if it has a model, and an axiom~$C\sqsubseteq D$ is \emph{entailed} by~\Omc (written $\Omc\models C\sqsubseteq D$) if every model of~\Omc is a model of~$C\sqsubseteq D$; in this case, we also say that $C$ is \emph{subsumed} by~$D$ w.r.t.\ \Omc.
The \emph{classification} of~\Omc is the set $\class(\Omc):=\{\tup{C,D}\mid C,D\in\sub(\Omc),\ \Omc\models C\sqsubseteq D\}$.\footnote{Often, the classification is done only for concept names in~\Omc, but we use a variant that considers all subconcepts, as it is done by the \ELbot reasoner~\Elk.}
% \stefan{I would prefer writing $C\sqsubseteq D$ instead of $\tup{C,D}$.}
%
The three reasoning problems of deciding consistency, checking subsumption, and computing the classification are mutually reducible in polynomial time.
Reasoning is \PTime-complete in \ELbot and \ExpTime-complete in \ALC~\cite{DBLP:books/daglib/0041477}.

%!TeX root=main.tex

\subsection{Concrete Domains.}
%\subsubsection*{Concrete Domains.}
%
Concrete domains 
have been introduced as a means to integrate reasoning about quantitative features of objects into DLs~\cite{DBLP:conf/ijcai/BaaderH91,DBLP:conf/aiml/Lutz02,DBLP:journals/jar/BaaderR22}.
Given a set \NPred of \emph{concrete predicates} and an arity $\ar(P)\in\Nbb$ for each $P\in\NPred$, a \emph{concrete domain (CD)} $\Dmc=(\Delta^\Dmc,\cdot^\Dmc)$ over $\NPred$ consists of a set $\Delta^\Dmc$ and relations $P^\Dmc\subseteq(\Delta^\Dmc)^{\ar(P)}$ for all $P\in\NPred$.
%
%For example, $\Delta^\Dmc$ could be the set \Qbb of rational numbers and \NPred could contain a binary predicate~$>$ interpreted as $\{(x,y)\in\Qbb^2\mid x>y\}$.
%For technical reasons, 
We assume that \NPred always contains a nullary predicate~$\bot$, interpreted as $\bot^\Dmc:=\emptyset$, and a unary predicate~$\top$ interpreted as $\top^\Dmc:=\Delta^\Dmc$.
Given a set \NV of \emph{variables}, a \emph{constraint} $P(x_1,\dots,x_{\ar(P)})$, with $P\in\NPred$ and $x_1,\dots,x_{\ar(P)}\in\NV$, is a predicate whose argument positions are filled with variables.

\begin{example}\label{cd:ex:1}
The concrete domain \DQgr has the set \Qbb of rational numbers as domain and, in addition to $\top$ and $\bot$, the concrete predicates $x=q$, $x>q$, and $x+q=y$, for constants $q\in\Qbb$, with their natural semantics~\cite{DBLP:conf/ijcai/BaaderBL05}. For example, $(x+q=y)^{\DQgr} = \{(p,r)\in \Qbb\times\Qbb \mid p+q=r\}$.\footnote{The index $\textit{diff}$ in its name is motivated by the fact that such a predicate fixes the difference between the values of two variables.}

The concrete domain \DQlin has the same domain as \DQgr, but its predicates other than $\{\top,\bot\}$ are given by linear equations $\sum_{i=1}^na_ix_i=b$, for $a_i,b\in \Qbb$, with the natural semantics~\cite{DBLP:journals/jar/BaaderR22},
%\footnote{In \cite{DBLP:journals/jar/BaaderR22}, the predicates of \DQlin are actually given by \emph{finite sets} of linear equations, which we can simulate by conjunctions of linear equations.} 
\eg the linear equation $x+y -z =0$ is interpreted as the ternary addition predicate 
%follows 
$(x+y -z =0)^{\DQlin} = \{(p,q,s)\in \Qbb^3 \mid p+q=s\}$.

The expressivity of these two CDs is orthogonal: The \DQgr predicate $x>q$ cannot be expressed as a conjunction of constraints in \DQlin, whereas the \DQlin predicate $x+y=0$ cannot be expressed in \DQgr.
\qed
\end{example}
A constraint $\alpha=P(x_1,\dots,x_{\ar(P)})$ is \emph{satisfied} by an assignment $v\colon\NV\to\Delta^\Dmc$ (written $v\models\alpha$) if $\big(v(x_1),\dots,v(x_{\ar(P)})\big)\in P^\Dmc$.
An \emph{implication} is of the form $\gamma\to\delta$, where $\gamma$ is 
a conjunction and $\delta$ a disjunction of constraints; 
%it is \emph{Horn} if $\delta$ contains a single constraint; 
it is \emph{valid} if all assignments satisfying all constraints in~$\gamma$ also satisfy some constraint in~$\delta$ (written $\Dmc\models\gamma\to\delta$).
A conjunction~$\gamma$ of constraints is \emph{satisfiable} if $\gamma\to\bot$ is not valid.
The CD~\Dmc is \emph{convex} if, for every valid implication 
$\gamma\to\delta$, there is a disjunct~$\alpha$ in~$\delta$ \st 
$\gamma\to\alpha$ is valid.
It is \emph{p-admissible} if it is convex and validity of implications is decidable in polynomial time.
This condition has been introduced with the goal of obtaining tractable 
extensions of \ELbot with 
concrete~domains~\cite{DBLP:conf/ijcai/BaaderBL05}.
%\stefan{@Franz: I do not discuss guarded convexity \cite{DBLP:journals/jar/BaaderR22} here since in our two concrete domains they coincide. OK?}
%\stefan{TODO: remove "Horn", other terms that are not needed?}

\begin{example}
The CDs \DQgr and \DQlin are both p-admissible, as shown in~\cite{DBLP:conf/ijcai/BaaderBL05} and~\cite{DBLP:journals/jar/BaaderR22}, respectively.
However, if we combined their predicates into a single CD, then we would lose convexity. In fact, \DQgr has the constraints $x>0$ and $x=0$. In addition, $y>0$ (of \DQgr) and $x+y=0$ (of \DQlin)
express $x<0$. Thus, the implication $x+y=0\rightarrow x>0 \vee x=0 \vee y>0$ is valid, but none of the implications $x+y=0 \rightarrow \alpha$ for $\alpha\in\{x>0,\ x=0,\ y>0\}$ is valid.
\qed
\end{example}

To integrate a concrete domain~\Dmc into description logics, the most general 
approach uses role paths $r_1\ldots r_k$ followed by a \emph{concrete feature} 
$f$ to instantiate the variables in constraints, where the $r_i$ are roles and $f$ is 
interpreted as a partial function $f^\Imc\colon\Delta^\Imc\to\Delta^\Dmc$. 
%from the abstract domain~$\Delta^\Imc$ to the concrete domain~$\Delta^\Dmc$. 
%
%\changeinpaper{For an interpretation \Imc, $u^\Imc = (r_1\ldots r_kf)^\Imc$ is defined as $\{(d,a) \subseteq \Delta^\Imc \times \Delta^\Dmc \ | \ \exists d_1,\dots,d_{k+1}\colon d = d_1, (d_i,d_{i+1})\in r_i^\Imc \text{ for } 1\leq i \leq k, \text{ and } f^\Imc(d_{i+1})=a \}$. Trivially, any concrete feature is a role path, \eg $u_1 = \ex{age}$, and role path $u_2 = \ex{parent}\,\ex{age}$ contains single role name $\ex{parent}$.}
%
Using the concrete domain \DQlin, the concept 
$
\ex{Human}\sqcap \exists \ex{age},\ex{parent}\,\ex{age}.[2x-y = 0]
$,
{for $\ex{age}$ being a concrete feature and $\ex{parent}$ a role name},
describes humans with a parent that has twice their age.\footnote{{See~\cite{DBLP:phd/dnb/Lutz02} for syntax and semantics of concepts using role paths.}} %\changeinpaper{Formally, semantics of concept $\exists u_1,u_2.P$, where $P=[2x-y = 0]$, is defined as $\{ d \in \Delta^\Imc \ | \ \text{there exist } a_1,a_2 \text{ with }  (d,a_i)\in u^\Imc_i \text{ for } i= 1,2,\text{ \st } 2a_1-a_2 = 0\}$}.
However, in the presence of role paths, p-admissibility of the CD does not 
guarantee decidability of the extended DL. %As shown 
%in~\cite{DBLP:journals/jar/BaaderR22}, 
Even if we just take the ternary addition predicate of \DQlin, the extension of \ALC with it becomes undecidable~\cite{DBLP:conf/cade/BaaderR20}, and the paper~\cite{DBLP:journals/jar/BaaderR22} exhibits a p-admissible CD whose integration into \ELbot destroys decidability.
%
%usually, one uses \emph{concrete features}~$f$ that are interpreted as partial functions $f^\Imc\colon\Delta^\Imc\to\Delta^\Dmc$ from the abstract domain~$\Delta^\Imc$ to the concrete domain~$\Delta^\Dmc$.
%
%Then, one can formulate \emph{concrete domain restrictions} like $\exists \ex{height},\ex{hasFather}\cdot\ex{height}.[x>y]$, where $x>y$ is a constraint, to express the concept of all people whose own $\ex{height}$ is larger than that of their fathers ($\ex{hasFather}\cdot\ex{height}$).
%
%Here, the \emph{role path} $\ex{hasFather}\cdot\ex{height}$ allows us to access the features of other abstract objects.
%
%Unfortunately, reasoning in DLs with concrete domains using role paths quickly becomes undecidable: for \ELbot, this happens even for the p-admissible concrete domain of all affine transformations over~$\Qbb^2$~\cite{DBLP:journals/jar/BaaderR22}, and \ALC becomes undecidable in combination with any concrete domain over \Qbb, \Zbb, or \Nbb that supports at least the constraints $x+1=y$ or $x+y=z$~\cite{DBLP:conf/cade/BaaderR20}.
%
Therefore, in this paper we disallow role paths, which effectively restricts concrete domain constraints to the feature values of single abstract objects. 
%\ie we can describe the set of all tall people by $\exists\ex{height}.[x>180]$, but not compare the heights of different people.
%
Under this restriction, the integration of a p-admissible CD leaves reasoning in 
\PTime for \ELbot~\cite{DBLP:conf/ijcai/BaaderBL05} and in \ExpTime for 
\ALC~\cite{DBLP:phd/dnb/Lutz02}.\footnote{%
The result in~\cite{DBLP:phd/dnb/Lutz02} applies to p-admissible CDs \Dmc since it is easy to show that the extension of \Dmc with the negation of its predicates satisfies the required conditions.
}
Disallowing role paths also enables us to simplify the syntax by treating variables directly as concrete features.
%, \ie writing $\ex{height}>180$ instead of the concept above, effectively making every constraint a concept.

Formally, the description logics \ELbotD and \ALCD are obtained from \ELbot and \ALC by allowing constraints~$\alpha$ from the CD~\Dmc to be used as concepts, where we employ the notation $[\alpha]$ to distinguish constraints visually from classical concepts.
Interpretations~\Imc are extended by associating to each variable $x\in\NV$ a \emph{partial} function $x^\Imc\colon\Delta^\Imc\to\Delta^\Dmc$, and defining $[\alpha]^\Imc$ as the set of all $d\in\Delta^\Imc$ for which (a)~the assignment $v^\Imc_d(x):=x^\Imc(d)$ is defined for all variables~$x$ occurring in~$\alpha$, and (b)~$v^\Imc_d\models\alpha$.
%
%Since $v_d^\Imc(x)$ may be undefined, for example, $\top\sqsubseteq[\top(x)]$ is not necessarily satisfied in every interpretation~\Imc.

%The semantics of concepts of the form $[\alpha]$ given above means that these concepts correspond to existential concrete domain restrictions without role paths. 
%In \ALCD, one usually also considers universal concrete domain restrictions $\forall[\alpha]$, denoting the set of all domain elements~$d\in\Delta^\Imc$ for which $v^\Imc_d\models\alpha$ holds whenever $v^\Imc_d$ is defined for all variables in~$\alpha$.
%%
%For $\alpha=P(x_1,\dots,x_k)$, this can be expressed in our syntax as $\lnot[\top(x_1)]\sqcup\dots\sqcup\lnot[\top(x_k)]\sqcup[\alpha]$~\cite{DBLP:conf/cade/HaarslevMW01}.
%%
%Similarly, in \ALCD one can express negated constraints $[\lnot\alpha]$ via $\lnot\forall[\alpha]$.
%

%\input{preliminaryexample}

\begin{example}\label{ex:medical}
Extending the medical example from the introduction, we can state that, for a patient in the intensive care unit, the heart rate and blood pressure are monitored,
%and thus these valus are always defined, 
using the GCI 
$
\ex{ICUpatient}\sqsubseteq [\top(\ex{hr})] \sqcap [\top(\ex{sys})] \sqcap [\top(\ex{dia})],
$
which says that, for all elements of the concept $\ex{ICUpatient}$, the values of the variables
$\ex{hr}$, $\ex{sys}$, $\ex{dia}$ are defined. The pulse pressure $\ex{pp}$ can then be defined via
$
\ex{ICUpatient}\sqsubseteq [\ex{sys}-\ex{dia} - \ex{pp} = 0].
$
Similarly, the maximal heart rate can be defined by
$
\ex{ICUpatient}\sqsubseteq[\ex{maxHR}+\ex{age}=220].
$
All the constraints employed in these GCIs are available in \DQlin.
One might now be tempted to use the GCI 
$
\ex{ICUpatient}\sqcap ([\ex{pp} > 50]\sqcup[\ex{hr} > \ex{maxHR}]) \sqsubseteq \ex{NeedAttention}
$
%and 
%$
%\ex{ICUpatient}\sqcap[\ex{hr} > \ex{maxHR}] \sqsubseteq \ex{NeedAttention}
%$ 
to say that ICU patients whose pulse pressure is larger than  50 mmHG or whose heart rate is larger than their maximal heart rate 
need attention. However, while 
%the first constraint 
{$[\ex{pp} > 50]$ }
is a \DQgr constraint, it is not available in \DQlin, and 
%the second one
{$[\ex{hr} > \ex{maxHR}]$ }
is available in neither. But we can raise an alert when the heart rate gets near the maximal one using
$[\ex{maxHR}-\ex{hr}=5]\sqsubseteq\ex{NeedAttention}$ { since it is a statement over \DQlin}.
% \stefan{Should we show a proof using these axioms?}
% \alisa{see appendix~\ref{appendix:proofpictures}}
\qed
\end{example}

%!TeX root=main.tex
\section{Combined Concrete and Abstract Reasoning}\label{sec:reasoning}

%We now consider reasoning in \ELbotD in more detail, but with our ultimate goal of computing proofs in mind.
% \stefan{I assume that proofs were introduced informally in the introduction.}
%
%An algorithm for classification in an extension of \ELbotD was first presented in~\cite{DBLP:conf/ijcai/BaaderBL05}; however, that algorithm has never been fully implemented.
%
%In fact, most existing DL reasoners do not support concrete domains, and those that do only allow \emph{unary} predicates~\cite{DBLP:journals/jar/GlimmHMSW14}, which is partially due to the fact that the standardized ontology language OWL~2 imposes this restriction~\cite{OWL}.
%
%On the other hand, \Elk~\cite{DBLP:journals/jar/KazakovKS14} is a highly optimized reasoner for \ELbot (without concrete domains) that is based on a set of rules that derive new GCIs from the ontology step-by-step, and is therefore ideally suited for producing proofs~\cite{DBLP:conf/dlog/KazakovKS17,DBLP:conf/cade/AlrabbaaBBDKM22}.
%
%To take advantage of this, in the following we describe an algorithm that uses \Elk as a black-box reasoner and introduces additional DL axioms to simulate the behavior of a p-admissible concrete domain with predicates of arities larger than~$1$.

We start by showing how classification in \ELbotD can be realized by interleaving a classifier for \ELbot with a constraint solver for \Dmc. Then we describe our constraint solvers
for \DQlin and \DQgr.

\subsection{Reasoning in \ELbotD}\label{subseq:el}
%\subsubsection*{Reasoning in \ELbotD.} \label{subseq:el}

The idea is that we can reduce reasoning in \ELbotD to reasoning in \ELbot by abstracting away CD constraints by new concept names, and then adding
GCIs that capture the interactions between constraints. To be more precise,
let \Dmc be a p-admissible concrete domain, \Omc an $\ELbotD$ ontology, and $\Cmc(\Omc)$ the finite set of constraints occurring in~\Omc.
We consider the ontology~$\Omc^{-\Dmc}$ that results from replacing each $\alpha\in\Cmc(\Omc)$ by a fresh concept name~$A_\alpha$.
Since \Dmc is p-admissible, the valid implications over the constraints in $\Cmc(\Omc)$ can then be fully encoded by the \ELbot ontology
\begin{align*}
  \Omc_\Dmc\coloneqq{}
  &\{A_{\alpha_1}\sqcap\dots\sqcap A_{\alpha_n}\sqsubseteq \bot \mid \alpha_1,\dots,\alpha_n\in\Cmc(\Omc),\ \Dmc\models \alpha_1\land\dots\land\alpha_n\to\bot\}\cup{} \\
  &\{A_{\alpha_1}\sqcap\dots\sqcap A_{\alpha_n}\sqsubseteq A_\beta \mid \alpha_1,\dots,\alpha_n,\beta\in\Cmc(\Omc),\ \Dmc\models \alpha_1\land\dots\land\alpha_n\to\beta\}.
\end{align*}
%\stefan{Simply say that we consider relevant GCIs of the forms ... that are implied by~\Dmc.}
%
The definition of $\Omc_\Dmc$ is an
adaptation of the construction introduced in \cite[Theorem~2.14]{DBLP:phd/dnb/Lutz02} for the more general case of admissible concrete domains. %to p-admissible concrete domains.
The problem is, however, that $\Omc_\Dmc$ is usually of exponential size since it considers all subsets $\{\alpha_1,\dots,\alpha_n\}$ of $\Cmc(\Omc)$. Thus, the reasoning procedure 
for $\ELbotD$ obtained by using $\Omc^{-\Dmc}\cup\Omc_\Dmc$ as an abstraction of~\Omc would also be exponential. To avoid this blow-up, we test implications of the form $\alpha_1\land\dots\land\alpha_n\to\bot$ and
$\alpha_1\land\dots\land\alpha_n\to\beta$ for validity in \Dmc only if this information is needed, i.e., if there is a concept $C$ that is subsumed by the concept names $A_{\alpha_1},\ldots,A_{\alpha_n}$.
%
%In practice, however, $\Omc_\Dmc$ can be very large, so our algorithm only considers those CD implications that are actually \emph{relevant} for reasoning over~\Omc.
%
%From inspection of the rules \textbf{CR7} and \textbf{CR8} from~\cite{DBLP:conf/ijcai/BaaderBL05}, one can see that the relevant implications $\alpha_1\land\dots\land\alpha_n\to\beta$ are those where a single concept~$C$ is subsumed by the conjunction $A_{\alpha_1}\sqcap\dots\sqcap A_{\alpha_n}$.
%\patrick{But now if I do not have ~\cite{DBLP:conf/ijcai/BaaderBL05} directly in front of me, I cannot understand this nor the intuition of the proof.}
%\stefan{If you don't want to check, then you have to trust that we are saying the truth, as with all citations. Should we explain it differently/in more detail?}

The resulting approach for classifying the $\ELbotD$ ontology \Omc, i.e., for computing $\class(\Omc)=\{\tup{C,D}\mid C,D\in\sub(\Omc),\ \Omc\models C\sqsubseteq D\}$ is described in 
Algorithm~\ref{alg:eld-reasoning}, where we assume that $\class(\Omc')$ is computed by a polynomial-time \ELbot classifier, such as \Elk, and that the validity of implications in \Dmc
is tested using an appropriate constraint solver for \Dmc. Since \Dmc is assumed to be p-admissible, there is a constraint solver that can perform the required tests in
polynomial time. Thus, we can show that this algorithm is sound and complete, and also runs in polynomial time.
%\patrick{Is it necessary to mention \Elk here? I would keep this generic, and just state that we would use a classifier that computes the required subsumption relationships and allows us to extract proofs. That we use \Elk in the end I would regard as 
%an implementation detail to be mentioned either as a remark at the beginning of this section, or later in the evaluation. For the 
%clarification on Lines~7 and~9, we could just state 
%
%}
%
%Since \Elk updates $\class(\Omc')$ incrementally whenever an axiom is added to~$\Omc'$ in Lines~\ref{l:update-bot} or~\ref{l:update-o}, it is more efficient to call \Elk to compute each set~$\Dbf_C$ in Line~\ref{l:dc}, instead of using the cached set~\Nbf from the previous iteration of the while loop.

\begin{algorithm}[tb]
  \caption{Classification algorithm for \ELbotD}\label{alg:eld-reasoning}
  $\Omc':=\Omc^{-\Dmc}$, $\Nbf:=\emptyset$\;\label{l:init}
  \While{$\Nbf\neq\class(\Omc')$}{
    $\Nbf:=\class(\Omc')$\;\label{l:update-n}
    \ForEach{$C\in \sub(\Omc^{-\Dmc})$}{
      $\Dbf_C:=\{\alpha\in\Cmc(\Omc)\mid \tup{C,A_\alpha}\in\class(\Omc')\}$\;\label{l:dc}
      \If{$\Dmc\models\bigwedge\Dbf_C\rightarrow\bot$}
        {$\Omc':=\Omc'\cup\big\{\bigsqcap_{\alpha\in\Dbf_C}A_\alpha\sqsubseteq\bot\big\}$\label{l:update-bot}}
      \Else{$\Omc':=\Omc'\cup\big\{\bigsqcap_{\alpha\in\Dbf_C}A_\alpha\sqsubseteq A_\beta\mid
        \beta\in\Cmc(\Omc), \Dmc\models\bigwedge\Dbf_C\rightarrow \beta\big\}$\label{l:update-o}}
    }
  }
  \Return{$\Nbf[A_\alpha\mapsto\alpha\mid\alpha\in\Cmc(\Omc)]$}
\end{algorithm}

\begin{restatable}{theorem}{ThmELDReasoning}\label{thm:eld-reasoning}
  Algorithm~\ref{alg:eld-reasoning} computes $\class(\Omc)$ in polynomial time.
\end{restatable}
%
%Intuitively, completeness of Algorithm~\ref{alg:eld-reasoning} can be shown by demonstrating that the axioms $\bigsqcap_{\alpha\in\Dbf_C}A_\alpha\sqsubseteq\bot$ and $\bigsqcap_{\alpha\in\Dbf_C}A_\alpha\sqsubseteq A_\beta$ mimic the CD reasoning rules \textbf{CR7} and \textbf{CR8} from~\cite{DBLP:conf/ijcai/BaaderBL05}.
%
%In the appendix, we provide a self-contained proof of this claim.

% In subsequent sections, we will describe how the algorithms for computing proofs and counter-interpretations using \Elk \cite{DBLP:conf/dlog/KazakovKS17,DBLP:conf/jist/AlrabbaaH22} can be combined with corresponding explanations for CD reasoning.

Next, we show how constraint solvers for \DQlin and \DQgr can be obtained.

%The only missing ingredient for \Cref{alg:eld-reasoning} is a reasoning procedure for deciding validity of the implications $\bigwedge\Dbf_C\to\bot$ and $\bigwedge\Dbf_C\to\beta$, which depends on the choice of concrete domain~\Dmc.
%
%In the following, we consider two known p-admissible concrete domains over the rational numbers, \ie where $\Delta^\Dmc=\Qbb$.

%\subsubsection*{Reasoning in~\DQlin.} % with Linear Equations.}
\subsection{Reasoning in~\DQlin} % with Linear Equations.}

To decide whether a finite conjunction of linear equations is satisfiable or whether it implies another equation,
we can use Gaussian elimination~\cite{Turner1995GaussEW}, which iteratively eliminates variables from a set of linear constraints in order to solve them. 
Each elimination step consists of a choice of constraint~$\alpha$ that is used to eliminate a variable~$x_i$ from another constraint~$\gamma$ by adding a suitable
 multiple~$q\in\Qbb$ of~$\alpha$, such that, in the sum $\gamma+q\alpha$, the coefficient~$a_i$ of~$x_i$ becomes~$0$.
This can be used to eliminate~$x_i$ from all constraints except~$\alpha$, which can then be discarded to obtain a system of constraints with one less variable.
For example, using $\alpha\colon 2x+3y=5$ to eliminate~$x$ from $\gamma\colon 4x-6y=1$ using $q=-2$ yields the new equation $-12y=-9$.

To decide whether $\alpha_1\land\dots\land\alpha_n\to\bot$  is valid in \DQlin, we must test whether the system of linear equations $\alpha_1,\dots,\alpha_n$ is unsolvable.
For this, we apply Gaussian elimination to this system. If we obtain a constraint of the form $0=b$ for non-zero $b$, then the system is unsolvable; otherwise, we obtain $0=0$
after all variables have been eliminated, which shows solvability. In case $\alpha_1\land\dots\land\alpha_n\to\bot$  is not valid, \Cref{alg:eld-reasoning} requires us to test whether 
$\alpha_1\land\dots\land\alpha_n\to\beta$ is valid for constraints $\beta$ different from $\bot$. This is the case iff the equation $\beta$ is a linear combination of the
equations $\alpha_1,\dots,\alpha_n$. For this, we can also apply Gaussian elimination steps to eliminate all variables from~$\beta$ using the equations $\alpha_1,\dots,\alpha_n$.
If this results in the constraint $0=0$, it demonstrates that $\beta$ is a linear combination; otherwise, it is not.

%In~\DQlin, an implication $\bigwedge\Dbf\to\beta$ is valid iff \Dbf is unsatisfiable or $\beta$ is a linear combination of the constraints in~\Dbf.
%
%If~$\Dbf$ is unsatisfiable, then repeated Gaussian elimination steps will yield a constraint of the form $0=b$ with $b\neq 0$; otherwise, the only remaining constraints, after all variables have been eliminated, must be of the form $0=0$.
%
%To check whether $\beta$ is a linear combination of the constraints in~\Dbf, we can similarly apply Gaussian elimination steps to eliminate all variables from~$\beta$ using the equations in~\Dbf---if this results in the constraint $0=0$, it demonstrates that $\beta$ is a linear combination of constraints from~\Dbf.
%
%Hence, we can use Gaussian elimination as a decision procedure for \DQlin in \Cref{alg:eld-reasoning}.

In principle, one could use standard libraries from linear algebra (\eg for Gaussian elimination or the simplex method~\cite{Turner1995GaussEW,DBLP:conf/cav/DutertreM06,DBLP:conf/tacas/MouraB08})
to implement a constraint solver for \DQlin.
%one can use standard methods and libraries from linear algebra, \eg Gaussian elimination or the simplex method~\cite{Turner1995GaussEW,DBLP:conf/cav/DutertreM06,DBLP:conf/tacas/MouraB08}.
We decided to create our own implementation based on Gaussian elimination, mainly for two reasons.
First, most existing numerical libraries are optimized for performance and use floating-point arithmetic. Hence, the results may be erroneous due to repeated rounding~\cite{DBLP:journals/computing/BarlowB85a}.
Second, even if rational arithmetic with arbitrary precision is used~\cite{DBLP:conf/cav/DutertreM06}, it is not trivial to extract from these tools a step-by-step account of how the verdict (valid or not) was obtained, which is a crucial requirement for extracting proofs.
%Therefore, we implemented a reasoning algorithm for~\DQlin that uses rational arithmetic and Gaussian elimination~\cite{Turner1995GaussEW}, which iteratively eliminates variables from a set of linear constraints in order to solve them.

%\subsubsection*{Reasoning in \DQgr.}
\subsection{Reasoning in \DQgr}

%The second p-admissible CD, called~$\DQgr$, supports constraints of the form $\top(x)$, $\bot$, $x=q$, $x>q$, and $x+q=y$, for constants $q\in\Qbb$, with their natural semantics~\cite{DBLP:conf/ijcai/BaaderBL05}.\footnote{We extend the CD from~\cite{DBLP:conf/ijcai/BaaderBL05} by $\bot$, which does not affect p-admissibility.}
%
The constraints of \DQgr can in principle be simulated in \emph{difference logic}, which consists of Boolean combinations of expressions of the form $x-y\le q$, and for which reasoning can be done using the Bellman-Ford algorithm for detecting negative cycles~\cite{DBLP:series/txtcs/KroeningS16,DBLP:conf/sat/ArmandoCGM04a}.
%
%For example, the constraints $x-2=y$, $y+3=x$ are unsatisfiable since there is a negative cycle between the difference constraints $x-y\le 2$, $y-x\le -3$, when viewed as a weighted graph with edges $x\xrightarrow{2}y$, $y\xrightarrow{-3}x$.
%
%This algorithm could be used to extract proofs for the unsatisfiability of a set of constraints~\Dbf (via a sequence of constraints that form a negative cycle), but no direct proofs for implications $\Dbf\to\beta$.
However, it is again not clear how proofs for the validity of implications can be extracted from the run of such a solver.
For this reason, we implemented a simple saturation procedure that uses the 
rules in Fig.~\ref{fig:cd2-rules} to derive implied constraints, where side 
conditions are shown in gray; these rules are similar to the rewrite rules for 
DL-Lite queries with CDs in~\cite{DBLP:conf/gcai/AlrabbaaKT19}.
\begin{figure}[tb]
  \begin{framed}
    \centering
    \def\defaultHypSeparation{\hskip 0em}

    \scalebox{.98}{
        \AXC{$x=q$}
        \AXC{$x=p$}
%        \AXC{\textcolor{gray}{$(q\neq p)$}}
        \RL{\Rnotequal\textcolor{black!40!white}{$\colon q\neq p$}}
        \BIC{$\bot$}
        \DP
    }
    \hfil
    \scalebox{.98}{
    \AXC{$x+q=y$}
    \AXC{$y+p=z$}
    \RL{\Rtrans}
    \BIC{$x+(q+p)=z$}
    \DP
}
    \hfil
    \scalebox{.98}{
        \AXC{\vphantom{$x=x$}}
        \RL{\Rzero}
        \UIC{$x+0=x$}
        \DP
    }

    \smallskip

    \scalebox{.98}{
        \AXC{$x+q=y$}
        \AXC{$x+p=y$}
%        \AXC{\textcolor{gray}{$(q\neq p)$}}
        \RL{\Rnotequalplus \textcolor{black!40!white}{$\colon q\neq p$}}
        \BIC{$\bot$}
        \DP
    }
    \hfil
    \scalebox{.98}{
        \AXC{$x=q$}
        \AXC{$y=p$}
        \RL{\Rdiff}
        \BIC{$x+(p-q)=y$}
        \DP
    }
    \hfil
    \scalebox{.98}{
        \AXC{$x+q=y$}
        \RL{\Rmirror}
        \UIC{$y+(-q)=x$}
        \DP
    }

    \smallskip

    \scalebox{.98}{
    \AXC{$x=q$}
    \AXC{$x>p$}
    %        \AXC{\textcolor{gray}{$(q<p)$}}
    \RL{\Rsmaller\textcolor{black!40!white}{$\colon q<p$}}
    \BIC{$\bot$}
    \DP
}
    \hfil
    \scalebox{.98}{
    \AXC{$x=q$}
    \AXC{$x+p=y$}
    \RL{\Requal}
    \BIC{$y=q+p$}
    \DP
}
\hfil
    \scalebox{.98}{
        \AXC{$x>q$}
        \AXC{$x+p=y$}
        \RL{\Rgreater}
        \BIC{$y>q+p$}
        \DP
    }
  \end{framed}
  \vskip-\bigskipamount
  \caption{Saturation rules for \DQgr constraints}
  \label{fig:cd2-rules}
\end{figure}
%
%Here, we do not explicitly consider $\top(x)$ since it can be expressed by $x+0=x$.
% \patrick{Looking at it in this presentation, it makes me think a lot of Christian's paper on query rewriting with CD. Maybe we should point this relation out.}
%
% Hence, we only deal with constraints of the forms $x=q$, $x>q$, $x+q=y$, and~$\bot$.
%
We eagerly apply the rules \Rnotequal, \Rsmaller, and \Rnotequalplus, which means that we only need to keep one constraint of the form $x+q=y$ in memory, for each pair $(x,y)$.
Since $x>q$ implies $x>p$ for all $p<q$, it similarly suffices to remember one unary constraint of the form $x=q$ or $x>q$ for each variable~$x$.
Apart from the three rules deriving~$\bot$, we can prioritize rules in the order \Rdiff, \Rmirror, \Rzero, \Rtrans, \Requal, \Rgreater, since none of the later rules can enable the applications of earlier rules to derive new constraints.
The full decision procedure is described in Algorithm~\ref{alg:cd2-reasoning}.

\begin{algorithm}[tb]
  \caption{Reasoning algorithm for \DQgr}\label{alg:cd2-reasoning}
  \KwIn{An implication $\bigwedge\Dbf\to\beta$ in \DQgr}
  \KwOut{\true iff $\DQgr\models\bigwedge\Dbf\to\beta$}
  $\Dbf':=\saturate(\Dbf)$\;
  \lIf{$\bot\in\Dbf'$ $\mathbf{or}$ $\beta\in\Dbf'$}{\Return{\true}}\label{l:bot}
  \If{$\beta$ is $x>q$}{
    \lIf{$x=p\in\Dbf'$ with $p>q$}{\Return{\true}}\label{l:constant-bound}
    \lIf{$x>p\in\Dbf'$ with $p\ge q$}{\Return{\true}}\label{l:weaken}
  }
  \Return{\false}
\end{algorithm}

\begin{restatable}{theorem}{ThmCDTwoReasoning}\label{thm:cd2-reasoning}
  \Cref{alg:cd2-reasoning} terminates in time polynomial in the size of $\bigwedge\Dbf\to\beta$ and returns \true iff $\DQgr\models\bigwedge\Dbf\to\beta$.
\end{restatable}

%!TeX root=main.tex
\section{Proofs for \ELbotD Entailments}\label{sec:proofs}

Our goal is now to use the procedures described in Section~\ref{sec:reasoning} to obtain separate proofs
for the DL part and the CD part of an entailment, which we then want to combine
into a single proof, as illustrated in Fig.~\ref{fig:proofs}.
% Having described the reasoning algorithms, we want to use them to compute explanations for entailments in \ELbotD, in the form of proofs~\cite{DBLP:conf/lpar/AlrabbaaBBKK20,DBLP:conf/cade/AlrabbaaBBKK21}.

\begin{figure}[tb]
 \input{figure-proofs}
 \caption{(a) \ELbot proof over $\Omc'$, (b) \DQlin proof and (c)
integrated $\ELbot[\DQlin]$ proof.}\label{fig:proofs}
\end{figure}

Fig.~\ref{fig:proofs}(a) shows an example of {an \Elk-proof}, a proof generated by the $\Elk$
reasoner~\cite{DBLP:conf/dlog/KazakovKS17} for the final ontology $\Omc'\supseteq\Omc^{-\Dmc}$ from \Cref{alg:eld-reasoning}. 
The
labels $\mathsf{R}_\sqsubseteq$ and $\mathsf{R}_\sqcap^+$ indicate the rules from the internal calculus of
\Elk~\cite{DBLP:journals/jar/KazakovKS14}, and $(*)$ marks an axiom added by
\Cref{alg:eld-reasoning}, where $\alpha$ is $2x+3y=5$, $\beta$ is $4y=3$, and $\gamma$ is $4x-6y=1$. 
We now describe how to obtain the proof~(b) for the CD implication $\alpha\land\beta\to\gamma$, and
how to integrate both proofs into the \ELbotD[\DQlin] proof~(c).

\subsection{Proofs for the Concrete Domains}
\label{sec:proofs-cd}

For \DQgr, 
%we can 
%use the algorithms
%from~\cite{DBLP:conf/lpar/AlrabbaaBBKK20,DBLP:conf/cade/AlrabbaaBBKK21}
%to extract proofs from the rule instances used in the saturation in
%Fig.~\ref{alg:cd2-reasoning} (see Fig.~\ref{fig:cd2-rules}).
the saturation rules in Fig.~\ref{fig:cd2-rules} can be seen as proof steps. Thus, the algorithms in~\cite{DBLP:conf/lpar/AlrabbaaBBKK20,DBLP:conf/cade/AlrabbaaBBKK21} can easily be adapted to extract $\DQgr$ proofs.
Inferences due to Lines~\ref{l:bot},~\ref{l:constant-bound} and~\ref{l:weaken} in
Algorithm~\ref{alg:cd2-reasoning} are captured by the following additional rules:
\[
  \AXC{$\bot$}
  \RL{\Rbot}
  \UIC{$\beta$}
  \DP
  \qquad
  \AXC{$x=p$}
   \RL{\Rconstantbound \textcolor{black!40!white}{$\colon p>q$}}
  \UIC{$x>q$}
  \DP
  \qquad
  \AXC{$x>p$}
  \RL{\Rweaken \textcolor{black!40!white}{$\colon p\ge q$}}
  \UIC{$x>q$}
  \DP
\]

For \DQlin, inferences are Gaussian elimination steps that derive
{$\sigma+c\rho$} from linear constraints~$\sigma$ and~$\rho$, and we label them with $[1,c]$ 
to indicate that $\sigma$ is
multiplied by $1$ and $\rho$ by $c$.
% by~$1$ and~$\alpha$ by~$c$ before they are added together.
%
This directly gives us a proof if the conclusion is~$\bot$ (or, equivalently, $0=b$ for non-zero~$b$).
However, proofs for implications $\bigwedge\Dbf\to\gamma$ need to be
treated differently. %
The Gaussian method would use 
%Since we use 
\Dbf to eliminate the variables from~$\gamma$ to show that $\gamma$
is a linear combination of~\Dbf, %, our approach 
and would yield a rather uninformative
proof with final conclusion $0=0$.
To obtain a proof with
$\gamma$
as conclusion, we reverse the proof direction by recursively applying the
following transformation starting from an inference step that has $\gamma$ as a premise:
% to make a  premise~$\alpha_1$ into a conclusion.
%\[ 
%  \AXC{$\alpha_1$}
%  \AXC{$\alpha_2$}
%  \RL{\scriptsize$[c,d]$}
%  \BIC{$\alpha_3$}
%  \DP
%  \qquad\leadsto\qquad
%  \AXC{$\alpha_2$}
%  \AXC{$\alpha_3$}
%  \RL{\scriptsize$[-\frac{d}{c},\frac{1}{c}]$}
%  \BIC{$\alpha_1$}
%  \DP
%\]

%\[ 
%  \AXC{$\alpha_1$}
%  \AXC{$\alpha_2$}
%  \RL{\scriptsize$[1,c]$}
%  \BIC{$\alpha_3$}
%  \DP
%  \qquad\leadsto\qquad
%  \AXC{$\alpha_2$}
%  \AXC{$\alpha_3$}
%  \RL{\scriptsize$[-c,1]$}
%  \BIC{$\alpha_1$}
%  \DP
%\]
\begin{equation}
  \AXC{$\sigma$}
  \AXC{$\rho$}
  \RL{\scriptsize$[1,c]$}
  \BIC{$\tau$}
  \DP
  \qquad\leadsto\qquad
  \AXC{$\rho$}
  \AXC{$\tau$}
  \RL{\scriptsize$[-c,1]$}
  \BIC{$\sigma$}
  \DP
  \label{transf}
  \tag{$\dagger$}
\end{equation}
Then we transform the next inference to obtain an inference that has~$\tau$ as the conclusion, and continue this process until $0=0$ becomes a leaf, which we then remove from the proof.
%
%Instead of the proof in
%\Cref{fig:proofs}(b), we would thus get

In our example, we would start with the following \enquote{proof} for 
$\Dmc\models\alpha\land\beta\to\gamma$:
\[
  \AXC{$4x-6y=1$}
  \AXC{$2x+3y=5$}
  \RL{\scriptsize$[1,-2]$}
  \BIC{$-12y=-9$}
  \AXC{$4y=3$}
  \RL{\scriptsize$[1,3]$}
  \BIC{$0=0$}
  \DP
\]
After applying two transformation steps~\eqref{transf}, we obtain the proof in Fig.~\ref{fig:proofs}(b).
%\changeinpaper{In particular, we obtain Fig.~\ref{fig:proofs}(b) by applying the transformations as above to 
%\[
%  \AXC{$4x-6y=1$}
%  \AXC{$2x+3y=5$}
%  \RL{\scriptsize$[1,-2]$}
%  \BIC{$-12y=-9$}
%  \AXC{$4y=3$}
%  \RL{\scriptsize$[1,3]$}
%  \BIC{$0=0$}
%  \DP
%\]
%}
\subsection{Combining the Proofs}\label{sec:combining-proofs}

It remains to integrate the concrete domain proofs into the DL proof over
$\Omc^{-\Dmc}$.
% Consider the proof in \Cref{fig:proofs}(b).
As a consequence
of \Cref{alg:eld-reasoning}, in~Fig.~\ref{fig:proofs}(a), the introduced concept 
names
$A_\alpha$, $A_\beta$, $A_\gamma$ occur in axioms with the same left-hand side $C$.
The idea is to add this \emph{DL context}~$C$ to every step
of the CD proof~(b) to obtain the \ELbotD-proof~(c).
This proof replaces the applications of
$\mathsf{R}_\sqcap^+$ 
%(building conjunction)
and $\mathsf{R}_\sqsubseteq$
%(using introduced axiom)
in the original DL proof~(a), and both (a) and (c) have essentially the same leafs and conclusion, except that
the auxiliary concept names~$A_\alpha,A_\beta,A_\gamma$ were replaced by the original constraints and the auxiliary axiom~$(*)$ was eliminated.
%The resulting proof has the same sink and leafs as the original
%\ELbot-proof~\eqref{p:elbot}(a), except that the additional axiom
%$A_\alpha\sqcap A_\beta\sqsubseteq\bot$ has been replaced by the inference
%labeled by $[-1,1]$.
%
% In our example, we obtain the proof in~\Cref{fig:proofs}(c).
%
In general, such proofs can be obtained by simple post-processing of proofs
obtained separately from the DL and CD reasoning components, and we conjecture
that the integrated proof~(c) is easier to understand in practice than the
separate proofs~(a) and~(b), since the connection between
the DL and CD contexts is shown in all steps.

\begin{restatable}{lemma}{LemTransformation}\label{lem:transformation}
  Let $\Omc'$ be the final ontology computed in \Cref{alg:eld-reasoning}. Given an \Elk-proof~$\p'$ for $\Omc'\models C^{-\Dmc}\sqsubseteq D^{-\Dmc}$ and proofs for all \Dmc-implications $\alpha_1\land\dots\land\alpha_n\to\beta$ used in~$\p'$, we can construct in polynomial time an $\ELbotD$-proof for $\Omc\models C\sqsubseteq D$.
\end{restatable}

\section{Generating Proofs for \ALCD}\label{sec:alc}
% \alisa{in contrast to the previous sections, this section includes ABoxes}
% \patrick{But in a very restricted form (only one individual), so maybe it is okay. We could also use a different name to avoid this confusion.}
\newcommand{\quant}{\mathsf{Q}}

For \ALCD, a black-box algorithm as for \ELbotD is not feasible, even though we consider only p-admissible concrete domains 
and 
no role paths. The intuitive reason is that \ALC itself is not convex, and we cannot simply 
use the classification result to determine which implications $\alpha_1\wedge\ldots\wedge\alpha_n\rightarrow\beta$ 
in~\Dmc are relevant. On the other hand, adding 
all valid implications is not practical, as there can be exponentially many.
We thus need a glass-box approach, \ie a modified \ALC reasoning procedure that determines the relevant CD implications on-demand.
% Modifying the reasoning procedure directly is also necessary 
% to trace inferences for our proofs, which is not directly supported by existing \ALC reasoners.

Moreover, to obtain proofs for \ALCD, we need a reasoning procedure that derives new axioms from old ones, and thus
classical tableau methods~\cite{ALC_TABLEAUX,HYPER_TABLEAUX} are not suited. 
However, existing consequence-based classification
methods for \ALC~\cite{DBLP:conf/ijcai/SimancikKH11} use complicated calculi that are not needed for our purposes. 
%\patrick{More references?}
Instead, we use a modified version of a calculus from~\cite{FORGETTING_ALC}, which uses only three inference rules, but 
performs double exponentially many inferences in the worst case. Our modification ensures that we perform at most 
exponentially many inferences, and are thus worst-case optimal for the \ExpTime-complete \ALCD. 
%With only three inference 
%rules, the calculus in~\cite{FORGETTING_ALC} is better suited, especially since we are not interested in classifying entire 
%ontologies. However, this calculus generates double-exponentially many inferences in the worst-case, while subsumption checking 
%in \ALCD is possible in \ExpTime. We therefore use a modification of the calculus in~\cite{FORGETTING_ALC}, which is optimal in 
%complexity, and still uses only on three inference rules. 
%
%For \ALC, there are no reasoners that can be used out-of-the-box to provide proofs,
%which is why we implemented a simple calculus loosely inspired by the
%refutationally complete calculus presented in~\cite{DBLP:conf/aaai/KoopmannS15}.%
%\footnote{The original calculus introduces fresh names during reasoning to
%represent conjunctions of concepts under role restrictions. This is not needed
%in our context, which allows our procedure to run in exponential time.}
%There are
%other calculi that support \ALC, such as the one used for the consequence-based
%reasoning approach in~\cite{DBLP:conf/ijcai/SimancikKH11}. We chose the one
%from~\cite{DBLP:conf/aaai/KoopmannS15} because it has a comparatively small and simple set
%of rules. We first present this calculus, and then describe how we integrate
%CD reasoning, for which we use a similar approach as for \ELbot.
%However, our 
%The procedure does not perform
%classification, but decides a specific entailment by a reduction to
%inconsistency checking.

\subsection{A Simple Resolution Calculus for \ALC}

The calculus represents GCIs $\top\sqsubseteq L_1\sqcup\dots\sqcup L_n$ as \emph{clauses} of the form 
\[
  L_1\sqcup\ldots\sqcup L_n \qquad  L_i\ ::=\ A\ \mid\ \neg A\ \mid\ \exists r.D\ \mid\ \forall r.D
\]
%where $x$ is understood as universally quantified variable, 
%and $C$ is a disjunction of \emph{concept literals}, which are 
%We assume the ontology to use only normalized concepts as follows.
%%
%\emph{Concept literals} are 
%concepts of the form $A$,
%$\neg A$, $\exists r.D$ and $\forall r.D$, 
where $n\geq 0$,
$A,D\in\NC$ and $r\in\NR$. %We represnet empty clauses (i=0) as $\bot$ for the case where $i=0$.
% and which represent GCIs $\top\sqsubseteq L_1\sqcup\ldots\sqcup L_n$.
%The TBox clause $x:C$ corresponds to the GCI $\top\sqsubseteq C$. 
%We also use \emph{ABox clauses}, which are of the form 
%$a: C$, with $C$ as for the TBox clause, and $a$ a specific constant referring to a domain element witnessing the (non-)subsumption 
%between two concept names. In particular, to prove $\Omc\models A\sqsubseteq B$, we normalize $\Omc$ into a set of TBox clauses,
%to which we add the two ABox clauses $a: A$, $a: \neg B$, and then apply our inference rules. No
%Having a unified form for ABox and TBox clauses helps keeping the calculus simple, as does the use of 
%disjunction as main connective. 
% GCIs involving complex concepts are reduced to GCIs between concept names as usual.
To decide 
$\Omc\models A\sqsubseteq B$, we normalize~$\Omc$ into a set of clauses, 
introducing fresh concept names for concepts under 
role restrictions, and add two special clauses $A_\text{LHS}\sqcup A$, $A_\text{RHS}\sqcup\neg B$, with fresh concept names~$A_\text{LHS}$
and~$A_\text{RHS}$. The latter are  used to track relevant inferences for 
constructing the final proof, for which we transform all clauses back 
into GCIs. 

\begin{figure}[t!]
 \begin{framed}
\scalebox{.85}{
  \textbf{A1}:
  $\dfrac{C_1\sqcup A,\quad C_2\sqcup\neg A}
         {C_1\sqcup C_2}$
    
    \qquad
  
  \textbf{r1}:
  $\dfrac{C\sqcup \exists r.D,\quad
          C_1\sqcup\forall r.D_1,\ \ldots,\
          C_n\sqcup\forall r.D_n,\quad \neg D_1\sqcup\ldots\sqcup \neg D_n}
         {C\sqcup C_1\sqcup\ldots\sqcup C_n}$
}

  \smallskip

\scalebox{.85}{
  \textbf{r2}:
  $\dfrac{C\sqcup \exists r.D,\quad
          C_1\sqcup\forall r.D_1,\ \ldots,\
          C_n\sqcup\forall r.D_n,\quad
          \neg D\sqcup \neg D_1\sqcup\ldots\sqcup \neg D_n}
         {C\sqcup C_1\sqcup\ldots\sqcup C_n}$
}

 \end{framed}
 \caption{Inference rules for \ALC clauses.}
 \label{fig:alc-calculus}
\end{figure}

Our inference rules are shown in Fig.~\ref{fig:alc-calculus}. 
%Here $t$ and $t_i$ stand for either $x$ or $a$, and 
%$\sigma(t_1,\ldots,t_n)$ is $a$ if $t_i=a$ for some $i$, and $x$ otherwise. 
\textbf{A1} is the standard resolution rule 
from first-order logic, which is responsible for direct inferences on concept names. The rules~\textbf{r1} and~\textbf{r2} perform 
inferences on role restrictions. They consider an existential role restriction~$\exists r.D$ and a (possibly empty) set
of value restrictions over~$r$, whose conjunction is unsatisfiable due to a clause over the nested concepts.
The concept~$D$ may not be relevant for this, which is why there are two rules.
Those rules are the main difference to the original calculus in~\cite{FORGETTING_ALC}, where a more expensive, incremental 
mechanism was used instead.
%
%this reasoning is performed 
%implicitly by propagating the concepts under value restrictions step-wise into the existential role restrictions, which 
%however leads to a worse-than-optimal reasoning complexity.
%
To transform this calculus into a practical method, we use optimizations common for resolution-based reasoning in 
first-order logic: ordered resolution, a set-of-support strategy, as well as backward and forward subsumption deletion. 
%Details can be found 
%in the
%\iftr {appendix.} \else {technical report.} \fi
In particular, our set-of-support strategy starts with a set of \emph{support clauses} containing only the clauses 
with $A_\text{LHS}$ and $A_\text{RHS}$. Inferences are always performed with at least one clause from this set, and the conclusion 
becomes a new support clause. If a support clause contains a literal $\exists r.D$/$\forall r.D$, 
we also add all clauses containing $\neg D$ as support clauses~\cite{ALC_ABDUCTION}.
% where we also show that the resulting method is sound and complete, and runs in exponential time.

%The denormalized
%version of the calculus, which is what would be shown to the user, can be found
%in the appendix.\todo{add this to the appendix}

\newcommand{\AlgALC}{$\textsf{Alg}_2$\xspace}
\newcommand{\AlgALCD}{$\textsf{Alg}_3$\xspace}

\subsection{Incorporating the Concrete Domain and Creating the Proof}

To incorporate concrete domains, we again work on the translation
$\Omc^{-\Dmc}$ replacing each constraint~$\alpha$ with~$A_\alpha$. In \ALCD, constraints can also
occur in negated form, which means that we can have literals $\neg A_\alpha$
expressing the negation of a constraint. 
%Our algorithm keeps track of 
%the set of currently relevant clauses, following the set-of-support strategy. This
We keep track of the set $\textbf{D}$ of concrete domain constraints $\alpha$ for which 
$A_\alpha$ occurs positively in a support clause.
%in turn leads to a set $\textbf{D}$ of currently relevant concrete domain constraints $D$ for which 
%the concept name $A_D$ occurs positively in the set of relevant clauses. 
We then use the proof procedure
for~\Dmc (see Section~\ref{sec:proofs-cd}) to generate all implications of the 
form $\alpha_1\wedge\ldots\wedge \alpha_n\rightarrow \beta$, where $\{\alpha_1$, $\ldots$, $\alpha_n\}\subseteq\mathbf{D}$ is subset-minimal, 
for which we add the corresponding clauses
$\neg A_{\alpha_1}\sqcup\ldots\sqcup\neg A_{\alpha_n}\sqcup A_\beta$. If $\beta=\bot$, we instead add   
$\neg A_{\alpha_1}\sqcup\ldots\sqcup\neg A_{\alpha_n}$. % In particular, using the 
%We denote the resulting reasoning 
%procedure \refstepcounter{algocf}{\textbf{Algorithm~\arabic{algocf}}\label{alg:alcd-reasoning}}:
%\begin{itemize}
%
%Details can be found in the
%\iftr {appendix.} \else {technical report.} \fi
%
%
%We add
%the following steps to our reasoning procedure to obtain 
%\refstepcounter{algocf}{\textbf{Algorithm~\arabic{algocf}}\label{alg:alcd-reasoning}}:
%\begin{itemize}
%  \item If there is some subset-minimal set $\Dbf$ of constraints for which $\bigwedge\Dbf$ is unsatisfiable and the literals $A_D$, $D\in\Dbf$,
%occur in some current clauses as maximal literals, then we add the clause $x:\bigsqcup_{D\in\Dbf}
%\neg A_{D}$.
%  \item If $\Dbf$ is a subset-minimal set of constraints that entails~$D$, and the literals $A_{D'}$, $D'\in\Dbf$, as well as $\lnot A_D$ occur as maximal literals in some current clauses, then we add
%$x:\bigsqcup_{D'\in\Dbf}\neg A_{D'}\sqcup A_D$.
%\end{itemize}
%\patrick{Doesn't work anymore since we don't define \enquote{maximal literals} in the main text.}

\begin{restatable}{theorem}{ThmALCDReasoning}\label{thm:alcd-reasoning}
  Let $\Omc$ be an \ALCD ontology and $\Nbf$ the
normalization of $\Omc^{-\Dmc}$. %, together with the clauses $a: A$, $a: \neg B$.
Then our method %running \Cref{alg:alcd-reasoning} on $\Nbf$ 
takes at most exponential time, and it derives
$A_\text{LHS}\sqcup A_\text{RHS}$ or a subclause from~\Nbf iff $\Omc\models C\sqsubseteq D$. % is unsatisfiable.
\end{restatable}
  Proofs generated using the calculus operate on the level of clauses.
  %If $\Omc\models C\sqsubseteq D$, then $x: A_\text{LHS}\sqcup A_\text{RHS}$ or a sub clause of it can be derived with our 
  %calculus, together with a corresponding proof from the set of clauses $\Nbf$. 
  We transform them into proofs of
  $\Omc^{-\Dmc}\models A\sqsubseteq B$ by
  1)~adding inference steps that reflect the normalization, 
  2)~if necessary, adding an inference to produce $A_\text{LHS}\sqcup A_\text{RHS}$ from a subclause
  3)~replacing $A_\text{LHS}$ by $\neg A$ and $A_\text{RHS}$ by $B$,
  4)~replacing all 
  other introduced concept names by the complex concepts they were introduced for, and
  5)~transforming clauses into more 
  human-readable GCIs using some simple 
  rewriting rules (%
\iftr {see the appendix} \else {see~\cite{ourarxive}} \fi
  for details). In the resulting proof, the initial clauses $A\sqcup A_{LHS}$ and $\neg B\sqcup A_{RHS}$
  then correspond to the tautologies $A\sqsubseteq A$ and $B\sqsubseteq B$.
  To get a proof for $\Omc\models A\sqsubseteq B$, we use a procedure similar to the one from Section~\ref{sec:combining-proofs} to integrate concrete domain 
  proofs.
  Because the integration requires only simple structural transformations,
the complexity of computing the combined proofs is determined by the
corresponding complexities for the DL and the
concrete domain.
We can thus extend the
approaches
from~\cite{DBLP:conf/lpar/AlrabbaaBBKK20,DBLP:conf/cade/AlrabbaaBBKK21}
to obtain complexity bounds for finding proofs of small size and depth.
% for formally investigating the complexity of finding good proofs also to the %
% case of DLs with CDs.
%
% % The goal is to find proofs of small \emph{size} (number of vertices) or \emph{depth} (length of the longest path from a leaf to the sink).

\begin{restatable}{theorem}{complexitygoodproofs}\label{thm:complexitygoodproofs}
  For $\Dmc\in\{\DQlin, \DQgr\}$, deciding the existence of a proof of at most
a given \emph{size} can be done in \NP for \ELbotD, and in \NExpTime for \ALCD.
For
 proof \emph{depth}, the corresponding problem is in \PTime for
$\ELbotD[\DQgr]$, in \NP for $\ELbotD[\DQlin]$, and in \ExpTime for $\ALCD$
(for both concrete domains).
%  \ExpTime for \ALCD, \NP for
% $\ELbot[\DQlin]$, and
\end{restatable}

\section{Implementation and Experiments}

%\stefan{@Christian: add technical details (Java, Elk, Lethe versions, CPU, memory, 
%timeouts etc.)}
We implemented the algorithms described above and evaluated their performance and the produced proofs on the self-created benchmarks \emph{Diet}, \emph{Artificial}, \emph{D-Sbj} and \emph{D-Obj}, each of which consists of multiple instances scaling from small to medium-sized ontologies. The latter two benchmarks are formulated in $\ELbotD[\DQgr]$, the rest in $\ELbotD[\DQlin]$.
Our tool is written using Java~8 and Scala. We used \Elk~0.5, \Lethe~0.85 and OWL 
API~4. The experiments were performed on 
Debian~Linux~10 (24 Intel Xeon~E5-2640 CPUs, 2.50GHz) with 25~GB 
maximum heap size and a timeout of~3~minutes for each task.  
Fig.~\ref{fig:experiments} shows the runtimes of the approaches for \ELbotD 
from Sections~\ref{sec:reasoning} and~\ref{sec:proofs} for reasoning and 
explanation depending on the \emph{problem size}, which counts all 
occurrences of concept names, role names, and features in the ontology.
%
% For each instance size, we computed the average runtime over three randomly generated instances of that size.
%
%Figure~\ref{fig:experiments} left shows explanation and reasoning time per benchmark instance.
%
% In addition, Figure~\ref{fig:classification} provides a comparison of the total reasoning time and the time used only by the calls to \Elk. 
% %
A more detailed description of the benchmarks and results can be found in \iftr {the appendix.} \else {\cite{ourarxive}.} \fi

\begin{figure}[tb]
\begin{minipage}{0.45\textwidth}
\centering
\includegraphics[width=\textwidth]{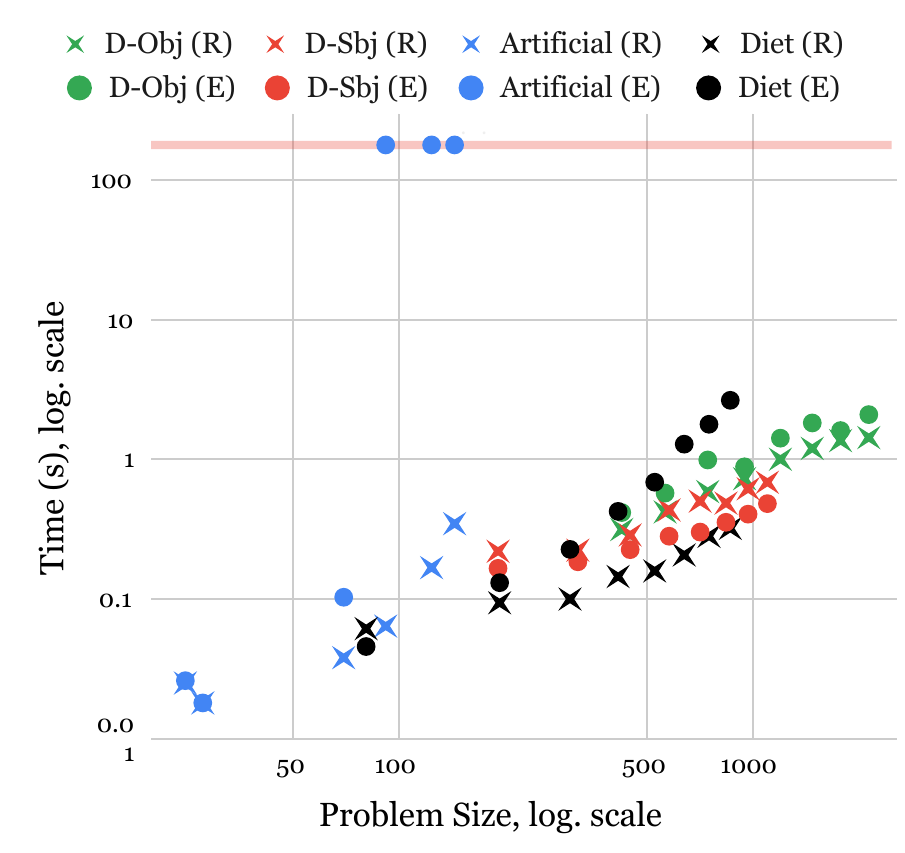}
\caption{\ELbotD: time for reasoning (R) and explanation (E) \vs problem size}
\label{fig:experiments}
\end{minipage}
\hfill
% \begin{minipage}[b]{0.45\textwidth}
% \centering
% \includegraphics[width=\textwidth]{pictures/class.pdf}
% \caption{Total reasoning time \vs \Elk runtime.}
% \label{fig:classification}
% \end{minipage}
\begin{minipage}{0.45\textwidth}
\centering
\includegraphics[width=\textwidth]{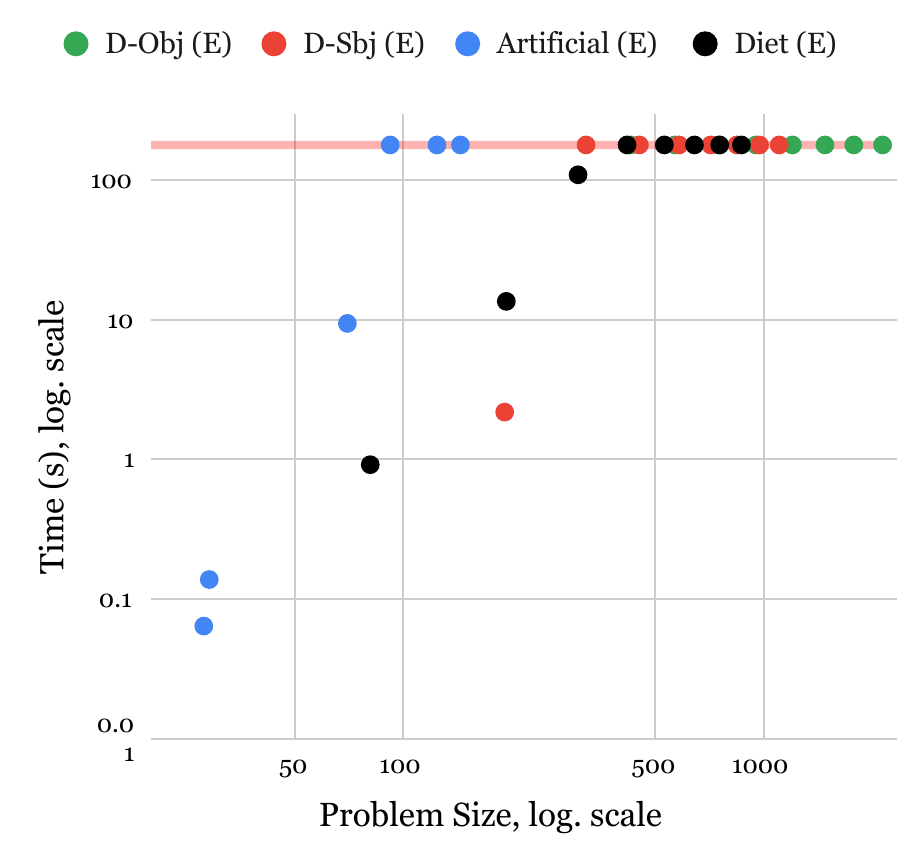}
\caption{\ALCD: total reasoning and explanation time \vs problem size}
\label{fig:experiments-alc}
\end{minipage}
\end{figure}

We observe that pure reasoning time (crosses in 
Fig.~\ref{fig:experiments}) scales well~\wrt problem size.
Producing proofs was generally more costly than reasoning, but the times were 
mostly reasonable.
%
% This can potentially be improved by optimizing the proof generation further.
%
However, there are several \emph{Artificial} instances for which the proof 
construction times out (blue dots). This is due to the nondeterministic choices of 
which linear constraints to use to eliminate the next variable, which we resolve 
using the Dijkstra-like algorithm described 
in~\cite{DBLP:conf/cade/AlrabbaaBBKK21},~which results in an exponential 
runtime in the worst case.
%
%We use the same algorithm to extract a single proof from the set of all 
%instantiated rules for \DQgr (see Figure~\ref{fig:cd2-rules}), which may also 
%grow quite large, \eg in the \emph{D-Obj} instances (green dots).
%
Another downside is that~some proofs were very large ($>2000$ inference steps 
in \emph{D-Obj}). However, we designed our benchmarks specifically to 
challenge the CD reasoning and proof generation capabilities (in particular, 
nearly all constraints in each ontology are necessary to entail the target axiom), 
and these results may improve for realistic ontologies.
% \stefan{in particular, we designed them such that we always need all constraints for entailing the target axiom}

Further analysis revealed that the reasoning times were often largely due to the 
calls to \Elk (ranging from $23\%$ in \emph{Diet} to $75\%$ in 
\emph{Artificial}), which shows that the CD reasoning does not add a huge 
overhead, unless the number of variables per constraint grows very large (\eg up 
to~$88$ in \emph{Diet}).
In comparison to the incremental use of \Elk as a black-box reasoner, the 
hypothetical \enquote{ideal} case of calling \Elk only once on the final saturated 
ontology~$\Omc'$ would not save a lot of time (average gain ranging from 
$42\%$ in \emph{Diet} to $14\%$ in \emph{D-Obj}), which shows that the 
incremental nature of our approach is also not a bottleneck.

Fig.~\ref{fig:experiments-alc} shows the runtime of the \ALCD calculus from Section~\ref{sec:alc}.
As expected, it performs worse than the dedicated \ELbotD algorithms.
In particular, currently there is a bottleneck for the \DQgr benchmarks (\emph{D-Sbj} and \emph{D-Obj}) that is due an inefficiency in the computation of the relevant CD implications $\alpha_1\land\dots\land\alpha_n\to\beta$.
In order to evaluate the increased expressivity supported by the \ALCD reasoner, 
we have also incorporated axioms with negation and universal restrictions into the 
\emph{Artificial} benchmark. Currently, however, the reasoner can solve only the 
smallest such instance before reaching the timeout.

We also compared our CD reasoning algorithms with Z3~\cite{DBLP:conf/tacas/MouraB08}, which supports linear arithmetic (for \DQlin) and difference logic (for \DQgr).
Ignoring the overhead stemming from the interface between Java and C++, the runtime of both approaches was generally in the same range, but our algorithms were faster on many CD reasoning problems.
This may be due to the fact that, although our algorithms for \DQlin and \DQgr 
are not optimized very much, they are nevertheless tailored towards very 
specific convex fragments: linear arithmetic with only~$=$, and difference logic 
with only $x+q=y$ and $x>q$, respectively.

\section{Conclusion}

We have shown that it is feasible to support p-admissible concrete domains in DL
reasoning algorithms, and even to produce integrated proofs for explaining 
consequences in the DLs \ELbotD and \ALCD, for the p-admissible concrete 
domains \DQlin and \DQgr.
In this work, we have restricted our attention to ontologies containing only GCIs (i.e., TBoxes) and to classification as the main reasoning problem. However, the extension of our methods to data and reasoning about individuals, \eg $\ex{fred}:\ex{ICUpatient}\sqcap[\ex{hr}=90]$,
encoded in so-called \emph{ABoxes}~\cite{DBLP:books/daglib/0041477}, is straightforward.
Likewise, the approach for computing \ELbotD proofs can be generalized to use other reasoning calculi for \ELbot %(e.g., the one in~\cite{DBLP:books/daglib/0041477}) 
instead of the one employed by \Elk, which makes very small proof steps and 
thus generates rather large proofs.

One major problem with using proofs to explain consequences is that they may 
become quite large. This problem already occurs for pure DLs without~CDs, and 
has also shown up in some of our benchmarks in this paper.
%, however, that the size of the proofs we obtain can be quite large, which makes them unsuitable as explanations for users of DL ontologies.
%
One possibility to alleviate this problem is to use an interactive proof visualization 
tool like Evonne~\cite{https://doi.org/10.1111/cgf.14730},
%~\cite{DBLP:conf/cade/AlrabbaaBBDKM22}, 
which allows zooming into parts of the proof and hiding uninteresting or already inspected parts. Since the integrated proofs that we generate have the same shape as pure DL proofs, they can be displayed using Evonne. It would, however, be interesting to add features tailored
to CD reasoning, such as visualizing the solution space of a system of linear equations.
%Therefore, in future work we will study how to reduce the size of these proofs by identifying the most important parts and trying to condense the rest.

In~\Cref{ex:medical}, we have seen that it would be useful to have the 
constraints of \DQlin and \DQgr available in a single CD.
Such a CD \Dmc would still preserve decidability if integrated into \ALC. However, since \Dmc is no longer convex, our reasoning approach for 
\ALCD does not apply. Thus, it would also be interesting to see whether this 
approach can be extended to \emph{admissible} 
CDs~\Dmc~\cite{DBLP:conf/ijcai/BaaderH91,DBLP:phd/dnb/Lutz02},
\ie CDs that are closed under negation and for which satisfiability of sets of 
constraints is decidable.

%On the theoretical side, we will also study the case of \ALCD with \emph{admissible} CDs~\Dmc~\cite{DBLP:phd/dnb/Lutz02}, \eg supporting the more interesting binary predicate~$x>y$.
% \stefan{@Patrick: How do you use the p-admissibility in the \ALCD algorithm? Would it be easy to adapt to admissible CDs?}

% \stefan{imprecise numbers/measurements! fuzzy concrete domains?}

 \paragraph{Acknowledgments}
   This work was supported by the DFG grant 389792660 as part of TRR~248
 (\url{https://perspicuous-computing.science}).
\bibliographystyle{plainurl}
\bibliography{bibs}

\iftr{
\newpage

\appendix

\section{Omitted Proofs in Sections~\ref{sec:reasoning} and~\ref{sec:proofs}}

% \section{Proofs for Section~\ref{sec:reasoning}}

\ThmELDReasoning*
\begin{proof}
We first observe that, in
each iteration, the while loop adds at most polynomially many axioms to $\Omc'$
(at most one for each concept $C\in\sub(\Omc^{-\Dmc})$ and constraint $\beta\in\Cmc(\Omc)$), which are of polynomial size.
Therefore, $\class(\Omc')$ can always be computed in polynomial time~\cite{DBLP:conf/ijcai/BaaderBL05} in Lines~\ref{l:update-n} and~\ref{l:dc}.
Moreover, there are at most polynomially many
iterations since we only add axioms to $\Omc'$, and thus due to the monotonicity of entailment, each set $\Dbf_C$
in Line~\ref{l:dc} monotonically increases from one iteration to the next, and is bounded by $\Cmc(\Omc)$.
Since the loop terminates once all sets $\Dbf_C$ remain the same,
there can be at most
$\lvert\sub(\Omc^{-\Dmc})\rvert \cdot\lvert\Cmc(\Omc)\rvert$ iterations of 
the while loop.

It remains to prove correctness. In each step,
$\Omc'[A_\alpha\mapsto\alpha\mid\alpha\in\Cmc(\Omc)]$ contains only axioms that are entailed
by $\Omc$, since the concept names~$A_\alpha$ replace exactly the occurrences of the concrete
constraints~$\alpha$ occurring in $\Omc$. Consequently, if the output of the
algorithm contains $\tup{C,D}$, then $\Omc\models C\sqsubseteq D$,
which means that the algorithm is sound.
It remains to show completeness,
\ie that the output contains all such tuples.
Take two concepts
$C^\dagger,D^\dagger\in\sub(\Omc)$ such that $\tup{C^\dagger, D^\dagger}$ is not returned. We construct a model $\Imc$ of $\Omc$, based
on the contents of~$\Omc'$ and $\Nbf=\class(\Omc')$ in the last iteration, such that $\Imc\not\models
C^\dagger\sqsubseteq D^\dagger$. We start with a model $\Imc'$ of the final ontology~$\Omc'$:
\begin{itemize}
 \item $\Delta^{\Imc'}:=\{C\in\sub(\Omc^{-\Dmc})\mid \tup{C,\bot}\not\in\Nbf\}$,
 \item for all $A\in\NC$, $A^{\Imc'}:=\{C\in\Delta^{\Imc'}\mid
\tup{C,A}\in\Nbf\}$,
 \item for all $r\in\NR$,
$r^{\Imc'}:=\{\tup{C,D}\in\Delta^{\Imc'}\times\Delta^{\Imc'}\mid
 \tup{C,\exists r.D}\in\Nbf\}$
\end{itemize}
Since $\class(\Omc')$ computes all entailed inclusions between
subconcepts in $\Omc'$, one can easily show that $\Imc'\models\Omc'$ \cite{DBLP:journals/jar/KazakovKS14}:
from the construction it follows by induction
on the structure of
$C,D\in\sub(\Omc^{-\Dmc})$ that $C\in D^{\Imc'}$ iff $\tup{C,D}\in\Nbf$ and
$\tup{C,\bot}\not\in\Nbf$. Let $C\sqsubseteq D\in\Omc'$ and $E\in C^{\Imc'}$.
Then, $\tup{E,C}\in\Nbf$, and since $\Omc'\models C\sqsubseteq D$, also
$\tup{E,D}\in\Nbf$, and thus $E\in D^{\Imc'}$.
% \stefan{@Patrick: The algorithm introduces new concepts $\bigsqcap_{D\in\Dbf_C}$ into~$\Omc'$. Don't we need domain elements for those as well, \ie we should use $\sub(\Omc')$ instead of $\sub(\Omc^{-\Dmc})$?}
% \patrick{It is not needed. We only want a model of $C^\dagger$, which is covered by this model.
% The construction further makes sure that every relevant role-successor is included. 
% The introduced subconcepts are not needed for the model. However, the classification of 
% $\Omc'$ makes sure that all relevant conjunctions over concrete predicates are still taken into account, so that 
% we can extend the model to a model of $\Omc$. 
% }
% \stefan{OK}

We argue that $\Imc'$ can be extended to an interpretation $\Imc$ such that, for each
$\alpha\in\Cmc(\Omc)$, we have $A_\alpha^{\Imc'}=\alpha^\Imc$.
It is possible
to ensure $A_\alpha^{\Imc'}\subseteq \alpha^\Imc$, because, if $\Dbf_C$ for some domain element $C\in\Delta^{\Imc'}$ is
unsatisfiable, we would have added an axiom
$\bigsqcap_{\alpha\in\Dbf_C}A_\alpha\sqsubseteq\bot$ in Line~\ref{l:update-bot}, and
by completeness of $\class(\Omc')$, we would have added $\tup{C,\bot}$ to $\Nbf$,
and thus not included $C$ in $\Delta^{\Imc'}$.
Consequently, we can assign values to the concrete features such that all
constraints in $\Dbf_C$ are satisfied, \ie $C\in A_\alpha^{\Imc'}$ implies $C\in \alpha^\Imc$, since $A_\alpha\in\Dbf_C$ by construction. However, we have to ensure that for all
$\alpha\in\Cmc(\Omc)$, also $\alpha^\Imc\subseteq A_\alpha^{\Imc'}$ holds, because otherwise there
might be GCIs in $\Omc$ with concrete constraints on the left-hand side that
are not satisfied. Intuitively, we have to make sure that we do not accidentally
satisfy more concrete domain constraints than necessary. For a fixed~$C$, let
$\overline{\Dbf}_C:=\{\beta\in\Cmc(\Omc)\mid\tup{C,A_\beta}\not\in\Nbf\}$.
Due to Line~\ref{l:update-o}, there is no $\beta\in\overline{\Dbf}_C$ such that
$\Dmc\models\bigwedge\Dbf_C\rightarrow\beta$. By convexity of~$\Dmc$, this implies
$\Dmc\not\models\bigwedge\Dbf_C\rightarrow\bigvee\overline{\Dbf}_C$.
%  which means
% that
% $\Dmc\not\models\left(\bigwedge\Dbf_C\wedge\bigwedge_{D\in\overline{\Dbf}_C}\neg
% D\right)\rightarrow\bot$.
In other words, we can find an assignment to the
concrete features $x^\Imc(C)$ such that every constraint in $\Dbf_C$ is satisfied, and no
constraint in $\overline{\Dbf}_C$ is satisfied. Doing this for all domain elements $C\in\Delta^{\Imc'}$,
we ensure that $A_\alpha^{\Imc'}=\alpha^\Imc$ for all $\alpha\in\Cmc(\Omc)$, and thus
$\Imc\models\Omc$.

We now show that $\Imc\not\models C^\dagger\sqsubseteq D^\dagger$ by proving that $C'\in(C')^\Imc=(C^\dagger)^\Imc$ and $C'\notin(D')^\Imc=(D^\dagger)^\Imc$, where $C'$ and $D'$ result from replacing in
$C^\dagger$ and $D^\dagger$ each $\alpha\in\Cmc(\Omc)$ by $A_\alpha$.
From our assumption that $\tup{C^\dagger,D^\dagger}$ is not returned by Algorithm~\ref{alg:eld-reasoning}, we know that $\tup{C',D'}\notin\Nbf$.
Therefore, by completeness of $\class(\Omc')=\Nbf$, we cannot have $\tup{C',\bot}\in\Nbf$, and thus $C'\in\Delta^{\Imc'}$ and $C'\notin(D')^{\Imc'}=(D')^\Imc$.
Completeness of $\class(\Omc')$ also yields that $\tup{C',C'}\in\Nbf$, and thus $C'\in(C')^{\Imc'}=(C')^\Imc$.
This shows that $\Imc\not\models C^\dagger\sqsubseteq D^\dagger$ and concludes the proof.
\qed
\end{proof}

\ThmCDTwoReasoning*
\begin{proof}
    We show that each rule can produce only quadratically many new constraints in the number of variables, and hence the algorithm terminates after polynomial time.
    For \Rdiff, \Rmirror, and \Rzero, this follows from the fact that these rules can be applied at most once for each variable~$x$ or each pair of variables~$(x,y)$.
    Using \Rtrans, we can produce at most~$2$ constraints of the form $x+q=y$ for each pair $(x,y)$, since the saturation stops as soon as~\Rnotequalplus can be applied.
    At the end, \Requal and \Rgreater can also be applied only once for each pair~$(x,y)$.
  
    It is clear that each of the rules is sound.
    Hence, it remains to show that $\DQgr\not\models\bigwedge\Dbf\to\beta$ whenever the Algorithm returns \false.
    In this case, $\Dbf'$ cannot contain~$\beta$ nor~$\bot$, \ie the rules \Rnotequalplus, \Rsmaller, and \Rnotequalplus are not applicable to~$\Dbf'$.
    We construct an assignment~$f$ to show that $\bigwedge\Dbf\to\beta$ is not valid, \ie which satisfies all constraints in~\Dbf, but not~$D\beta$.
    For any $x=q\in\Dbf'$, we set $f(x):=q$.
    Consider now the remaining unsatisfied constraints of the form $x+q=y$ and $x>q$ in~$\Dbf'$, for which there cannot exist constraints $x=p$ nor $y=p$ in $\Dbf'$.
    Due to \Rmirror and \Rtrans, the directed graph~$G$ with edges $\{(x,y)\mid x+q=y\in\Dbf'\}$ consists of unconnected cliques.
    If we fix one variable $x$ in each clique~$C$ and some value $f(x)$, then the values $f(y)$ of the other variables $y\in C$ are determined by the unique constraints $x+q=y$ that must exist in~$\Dbf'$.
    Such an assignment also satisfies all constraints $y+p=z$ with $y,z\in C$ since they are implied by corresponding constraints $x+q=y$ and $x+r=z$, for which we must have $r-q=p$ due to \Rmirror, \Rtrans, and \Rnotequalplus.
    We now set $f(x)$ to an arbitrary value that only has to satisfy the constraint $x>q$ in case one exists in~$\Dbf'$, and fix the values of all other $y\in C$ accordingly.
    All constraints $y>p\in\Dbf'$ for $y\in C$ will be satisfied due to~\Rgreater.
  
    If $\beta$ is either of the form $x+q=y$ or $x=q$, it could happen that we have chosen values $f(x),f(y)$ that satisfy~$\beta$ by accident.
    However, since we assumed that $\beta\notin\Dbf'$, this is only possible if there is no constraint $x+p=y$ or $x=p$, respectively, in $\Dbf'$.
    In particular, $x$ and $y$ cannot be in the same clique.
    Thus, we can increase the value $f(x)$ by an arbitrary amount (and the values of all variables connected to~$x$ in~$G$ accordingly) in order to ensure that $\beta$ is not satisfied, while $\Dbf\subseteq\Dbf'$ remains satisfied.
  
    If $\beta$ is of the form $x>q$, then we know that $\Dbf'$ contains neither $x=p$ with $p>q$ nor $x>p$ with $p\ge q$.
    If $\Dbf'$ contains $x=p$ with $p\leq q$, then $f(x)=p$ does not satisfy~$\beta$.
    If $\Dbf'$ contains $x>p$ with $p<q$, then it is possible to choose $f(x):=q$ (and adjust the values of $x$'s clique accordingly), which also does not satisfy~$\beta$, but still satisfies~\Dbf.
    Finally, if $\Dbf'$ contains no constraint of the form $x=p$ or $x>p$, the value of $f(x)$ can be chosen arbitrarily while still satisfying~$\Dbf$, and so we can again choose $f(x):=q$ (and adjust the connected variables accordingly).
    \qed
  \end{proof}

% \section{Proofs for Section~\ref{sec:proofs}}

\LemTransformation*
\begin{proof}
    Given a \Dmc-proof~\p and a concept~$C$ (the \emph{context}), we construct the \ELbotD-proof~$\p_C$ by replacing each inference in~\p as follows:
    \[
      \AXC{$\gamma_1$}
      \AXC{$\ldots$}
      \AXC{$\gamma_m$}
      \RL{$\mathsf{R}$}
      \TIC{$\delta$}
      \DP
      \quad
      \leadsto
      \quad
      \AXC{$C\sqsubseteq A_{\gamma_1}$}
      \AXC{$\ldots$}
      \AXC{$C\sqsubseteq A_{\gamma_m}$}
      \RL{$\mathsf{R}$}
      \TIC{$C\sqsubseteq A_{\delta}$}
      \DP
    \]
    In this transformation, all steps remain sound, and a \Dmc-proof of~$\beta$ becomes an \ELbot-proof of $C\sqsubseteq A_{\beta}$.
    
    To integrate this into the original \ELbot-proof~$\p'$ from \Elk, we need to analyze in which contexts the newly introduced axioms $A_{\alpha_1}\sqcap\dots\sqcap A_{\alpha_n}\sqsubseteq A_\beta$ can appear in~$\p'$ (the case with~$\bot$ instead of~$A_\beta$ is similar).
    The relevant inference rules of \Elk are
    \[
      \def\defaultHypSeparation{\hskip 0em}
      \AXC{$C\sqsubseteq D$}
      \AXC{$D\sqsubseteq E$}
      \RL{$\mathsf{R}_\sqsubseteq$}
      \BIC{$C\sqsubseteq E$}
      \DP
      \qquad
      \AXC{$C\sqsubseteq D$}
      \AXC{$C\sqsubseteq E$}
      \RL{$\mathsf{R}_\sqcap^+$}
      \BIC{$C\sqsubseteq D\sqcap E$}
      \DP
      \qquad
      \AXC{$C\sqsubseteq\exists r.D$}
      \AXC{$D\sqsubseteq E$}
      \RL{$\mathsf{R}_\exists$}
      \BIC{$C\sqsubseteq\exists r.E$}
      \DP
    \]
    with the side condition that the axiom $D\sqsubseteq E$ in~$\mathsf{R}_\sqsubseteq$ is always an element of the input ontology (here, $\Omc'$).
    We distinguish three cases in the following.
    
    \begin{itemize}
      \item In the best case, $A_{\alpha_1}\sqcap\dots\sqcap A_{\alpha_n}\sqsubseteq A_\beta$ is used in~$\p'$ in an inference step
      \begin{equation}
        \AXC{$C\sqsubseteq A_{\alpha_1}\sqcap\dots\sqcap A_{\alpha_n}$}
        \AXC{$A_{\alpha_1}\sqcap\dots\sqcap A_{\alpha_n}\sqsubseteq A_\beta$}
        \RL{$\mathsf{R}_\sqsubseteq$}
        \BIC{$C\sqsubseteq A_\beta$}
        \DP
        \label{step:r-sqsubseteq}
      \end{equation}
      as in our example proof in Fig.~\ref{fig:proofs}(a).
      Then, we replace~\eqref{step:r-sqsubseteq} by~$\p_C$, where \p is the \Dmc-proof of $\alpha_1\land\dots\land\alpha_n\to\beta$.
      The proof~$\p_C$ already has the same conclusion $C\sqsubseteq A_\beta$ as~\eqref{step:r-sqsubseteq}.
      However, since the leafs of~$\p_C$ are of the form $C\sqsubseteq A_{\alpha_i}$, in general we also need to add the inferences
      \begin{prooftree}
        \AXC{$C \sqsubseteq A_{\alpha_1}\sqcap \ldots\sqcap A_{\alpha_n}$}
        \RL{$\mathsf{R}_\sqcap^-$}
        \UIC{$C \sqsubseteq A_{\alpha_i}$}
      \end{prooftree}
      to~$\p'$ to connect~$\p_C$ to the original premise $C\sqsubseteq A_{\alpha_1}\sqcap\dots\sqcap A_{\alpha_n}$ of~\eqref{step:r-sqsubseteq}.
      In Fig.~\ref{fig:proofs}(c), this was not necessary since the nodes labeled by 
      $C\sqsubseteq A_\alpha$ and $C\sqsubseteq A_\beta$ already existed 
      in~(a), and hence we could use them directly and omit the 
      step~$\mathsf{R}_\sqcap^+$ in~(a).
    
      \item An axiom $A_{\alpha_1}\sqcap\dots\sqcap A_{\alpha_n}\sqsubseteq A_\beta$ could also be used in~$\mathsf{R}_\exists$:
      \begin{equation}
        \AXC{$C\sqsubseteq\exists r.\big(A_{\alpha_1}\sqcap\dots\sqcap A_{\alpha_n}\big)$}
        \AXC{$A_{\alpha_1}\sqcap\dots\sqcap A_{\alpha_n}\sqsubseteq A_\beta$}
        \RL{$\mathsf{R}_\exists$}
        \BIC{$C\sqsubseteq\exists r.A_\beta$}
        \DP
        \label{step:r-exists}
      \end{equation}
      In this case, we cannot use~$C$ as the context.
      Though it would be tempting to translate \Dmc-inferences 
      $\frac{\gamma_1~~\dots~~\gamma_m}{\delta}$ into $\frac{C\sqsubseteq\exists 
      r.A_{\gamma_1}~~\dots~~C\sqsubseteq \exists r.A_{\gamma_m}}{C\sqsubseteq\exists 
      r.A_\delta}$ to arrive at the conclusion $C\sqsubseteq\exists r.A_\beta$,
      % \christian{Shouldn't the premises have more than 2 elements?}\stefan{Better?} 
      such inferences would not be sound since $\emptyset\not\models\exists 
      r.A_{\gamma_1}\sqcap\dots\sqcap\exists r.A_{\gamma_m}\sqsubseteq \exists 
      r.(A_{\gamma_1}\sqcap\dots\sqcap A_{\gamma_m})$ in \ELbot.
      Hence, we simply use $D=A_{\alpha_1}\sqcap\dots\sqcap A_{\alpha_n}$ itself as 
      the context, \ie we translate the \Dmc-proof~\p of $\alpha_1\land\dots\land\alpha_n\to\beta$ into $\p_D$, which provides a proof of the second premise $D\sqsubseteq A_\beta$ of~\eqref{step:r-exists}.
      % \christian{Do you mean $D$ is the context in the inferences of $D\sqsubseteq 
      % E$? I thought the context in $\mathsf{R}_\exists$ is $C$ }
      % \stefan{We cannot use $C$ as the context for the CD implication here, as explained above. As far as I understood, in this case the implementation would fall back on using $D=A_{\alpha_1}\sqcap\dots\sqcap A_{\alpha_n}$ as the context, \ie as the left-hand side everywhere in the proof for $A_{\alpha_1}\sqcap\dots\sqcap A_{\alpha_n}\sqsubseteq A_\beta$, \eg $\frac{A_{\alpha_1}\sqcap\dots\sqcap A_{\alpha_n}\sqsubseteq A_{\gamma_1}~~A_{\alpha_1}\sqcap\dots\sqcap A_{\alpha_n}\sqsubseteq A_{\gamma_2}}{A_{\alpha_1}\sqcap\dots\sqcap A_{\alpha_n}\sqsubseteq A_{\delta}}$ for every \Dmc-step $\frac{\gamma_1~~\gamma_2}{\delta}$ etc.}
      %
      The leafs of~$\p_D$ have labels of the form $A_{\alpha_1}\sqcap\dots\sqcap A_{\alpha_n}\sqsubseteq A_{\alpha_i}$, for which we introduce new inferences $\frac{}{A_{\alpha_1}\sqcap\dots\sqcap A_{\alpha_n}\sqsubseteq A_{\alpha_i}}$ since they are tautologies that require no further explanation.
      Using $D=A_{\alpha_1}\sqcap\dots\sqcap A_{\alpha_n}$ as context here is not ideal, though, since in general $n$ could be quite large, which can make the proof cluttered.
    
      \item The last case is when $A_{\alpha_1}\sqcap\dots\sqcap A_{\alpha_n}\sqsubseteq A_\beta$ appears as the first premise of~$\mathsf{R}_\sqsubseteq$ or any premise of~$\mathsf{R}_\sqcap^+$.
      In this case, the left-hand side $C=A_{\alpha_1}\sqcap\dots\sqcap A_{\alpha_n}$ is propagated to the conclusions of the form $C\sqsubseteq E$ or $C\sqsubseteq D\sqcap E$, which then potentially lead to more axioms of the form $C\sqsubseteq F$.
      However, since these axioms are not elements of~$\Omc'$, they can never appear as the second premise of~$\mathsf{R}_\sqsubseteq$, which means that we cannot find a different DL context for the \Dmc-proof~\p of $A_{\alpha_1}\sqcap\dots\sqcap A_{\alpha_n}\sqsubseteq A_\beta$, and we again have to use $C=A_{\alpha_1}\sqcap\dots\sqcap A_{\alpha_n}$ itself as the context, \ie we use~$\p_C$ to derive $A_{\alpha_1}\sqcap\dots\sqcap A_{\alpha_n}\sqsubseteq A_\beta$ directly.
      %
      % In this case, we recursively consider how the conclusion $C\sqsubseteq E$ or 
      % $C\sqsubseteq D\sqcap E$ is used in subsequent inference steps, in order to 
      % find a DL context~$C'$ that we can use for the involved \Dmc-proofs.
      % \christian{In the implementation, we do not look into $C$ (at least for 
      % $\mathsf{R}_\sqsubseteq$).  I don't understand this.}
      % \stefan{This should be fixed now. Can you check again?}
      % %
      % If $C\sqsubseteq E$ or $C\sqsubseteq D\sqcap E$ is already the final conclusion of~$\p'$, then we have no other choice than to use $C=A_{\alpha_1}\sqcap\dots\sqcap A_{\alpha_n}$ as the context.
    
    \end{itemize}
    % \christian{only inferences of $\mathsf{R}_\sqsubseteq$ get replaced. Inferences 
    % $\mathsf{R}_\sqcap^+$ and $\mathsf{R}_\exists$ stay without modification. 
    % Why do we distinguish them?}
    % \stefan{We need to say what happens with axioms $A_{\alpha_1}\sqcap\dots\sqcap A_{\alpha_n}\sqsubseteq A_\beta$ if they are used in rules $\mathsf{R}_\sqcap^+$ or $\mathsf{R}_\exists$. We cannot just leave them in~$\p'$ as-is, since that would not be correct \wrt \Omc.}
    %
    Finally, as seen in Fig.~\ref{fig:proofs}(c), we replace all $A_\alpha$ in the resulting proof by~$\alpha$ in order to obtain an \ELbotD-proof.
    % \stefan{@Christian: Check if this corresponds to the implementation.}
    % \christian{In case the task is to explain $C \sqsubseteq A_\alpha$, replacing 
    % all $A_\alpha$ with $\alpha$ gives a proof with a mismatching sink. For these, 
    % we add additional inferences
    % $\frac{C \sqsubseteq \alpha \qquad \alpha \sqsubseteq A_\alpha}{C 
    % \sqsubseteq A_\alpha} \quad \mathsf{R}_\textsf{Norm}$ \qquad
    % $\frac{}{\alpha \sqsubseteq A_\alpha} \quad \mathsf{R}_\textsf{M}$
    % } 
    % \stefan{$C \sqsubseteq A_\alpha$ cannot be the task, since that is not an \ELbotD axiom. The task to explain has to be formulated over the original \ELbotD ontology~\Omc, so it would be $C\sqsubseteq[\alpha]$.\\
    % In the implementation, I assume we need to do this since there can be multiple concept names representing the same constraint~$\alpha$? But this is not possible here ($A_\alpha$ is the only representative of~$\alpha$), this is just an artifact of how it was implemented.}
    % \christian{Why is the format of $C$ is relevant for the proof?}
    % \stefan{What do you mean?}
    %
    \qed
    \end{proof}

\section{Omitted Details from \Cref{sec:alc}}

\subsection{Refutational Completeness}

% We assume that our normalization respects the following side condition:

%We directly present the version of the calculus that uses ordered resolution when showing
%correctness of the reasoning procedure.
%For efficiency reasons, we also assume a
We first show refutational completeness of the calculus for \ALC. The idea is to
present a procedure that constructs
a model from a saturated set of clauses that does not contain the empty clause.
The construction makes use of linear orderings on clauses and literals that determine how the
model is constructed. These orderings are also used in the optimized version of the calculus.
While this is a common construction to show refutational
completeness of resolution, the peculiarities of our method requires a more
refined ordering.

Note the special role of concept names occurring under a role restriction: if a
concept name $A$ occurs in a literal $\quant r.A$, then we assume that there
are no other positive occurrences of $A$, that is, literals with $A$ are either
of the form $\neg A$ or $\quant r.A$, where $\quant$ and $r$ are always the
same. Correspondingly, for every concept name $A$, positive occurrences of $A$
in the set of clauses are either only in literals of the form $A$, only in
literals of the form $\exists r.A$, or only in literals of the form $\forall
r.A$.

We now assume a linear order $\prec_l$ on literals such that, for all $A\in\NC$:
\begin{enumerate}
 \item\label{cond:exists} If $A_1$ occurs under an existential role restriction, and $A_2$ under a value restriction,
    then $\neg A_1\prec\neg A_2$.
 \item\label{cond:definers} If $A$ occurs under a role restriction, then $\neg A\prec_l L$ for any
literal $L$ that is not of the form $\neg A'$, where $A'$ occurs under a
role restriction.
 \item\label{cond:negative} If $A$ does not occur under a role restriction, then $A\prec_l\neg A$,
 \item\label{cond:quantifiers} $\exists r.A_1\prec_l\forall s.A_2\prec_l A_3$ for all $r,s\in\NR$ and
$A_1$, $A_2$, $A_3\in\NC$.
\end{enumerate}
% \stefan{Would this be correct? For all literals $\exists r.A_1$, $\forall s.B_2$, $A_3$ (not under role restriction), we have $\lnot A_1\prec_l\lnot A_2\prec_l\exists r.A_1\prec_l\forall s.A_2\prec_l A_1\prec_l A_2\prec_l A_3\prec_l\lnot A_3$, and the order among the different "groups" of literals, \eg different ones of the form $\lnot A_2$, is arbitrary.}

We extend $\prec_l$ to a linear order $\prec$ on clauses, such that $C_1\prec C_2$
whenever $C_2$ contains a literal $L_2$ such that, for every literal $L_1$ in $C_1$, $L_1\prec C_2$.

\begin{theorem}\label{the:alc-calculus}
  For any set of clauses $\Nbf$ satisfying our constraints, our calculus
  derives a number of clauses that is at most exponential in
$\Nbf$, and derives the empty clause iff $\Nbf$ is inconsistent.
\end{theorem}
\begin{proof}
  As usual, we will use the symbol $\bot$ to denote the empty clause.
  The number of distinct literals occurring in $\Nbf$ is linearly bounded by
its size, and each derived clause is composed of literals occurring in $\Nbf$.
Since we represent clauses as sets, this establishes that the algorithm has to
terminate after exponentially many steps. Every rule in Fig.~\ref{fig:alc-calculus}
infers only clauses that are entailed by its premises. Consequently, if
$\bot$ is derived, then $\Nbf$ must be inconsistent. It remains to show that
if $\bot$ is not derived, then $\Nbf$ has a model.

Let $\Nbf^*$ be the set of clauses generated by \AlgALC, and assume
$\bot\not\in\Nbf^*$.
%We extend $\prec_l$ to a total
%order $\prec$ between clauses s.t. $C_1\prec C_2$ whenever $C_2$ contains a
%literal $L_2$ s.t. $L_1\prec_l L_2$ for all literals $L_1$ in $C_1$.
Using the order $\prec$ on clauses,
we construct an interpretation $\Imc$ as fixpoint of an unbounded
sequence of interpretations $\Imc_0$, $\Imc_1$, $\ldots$. 
The first interpretation~$\Imc_0$ is defined
by $\Delta^{\Imc_0}=\{d_0\}$ and $X^{\Imc_0}=\emptyset$ for all
$X\in\NC\cup\NR$. The next interpretations are built based on the previous one,
where in each step, we add at most one element to the domain. We may thus speak
of the \emph{oldest domain element} (that satisfies some condition), by which we mean the
domain element added first in the sequence of interpretations.
$\Imc_{i+1}$ is constructed from $\Imc_{i}$ as follows. If
$\Imc_{i}\models\Nbf^*$, then $\Imc_i=\Imc_{i+1}=\Imc$. Otherwise, there
is a clause $C\in\Nbf^*$ and a domain element $d\in\Delta^{\Imc_i}$ s.t.
$d\not\in C^{\Imc_i}$.
%either
%is a clause of the form $C\in\Nbf^*$ s.t. $a^\Imc\not\in C^{\Imc_i}$, or a
%clause of the form $x: C\in\Nbf^*$ and some $d\in\Delta^{\Imc_i}$ s.t.
%$d\not\in C^{\Imc_i}$.
Choose such a pair $\tup{d,C}\in\Delta^{\Imc_i}\times\Nbf^*$
where, among all such
pairs, $d$ is the \emph{oldest} domain element (that is, the domain
element introduced in~$\Imc_j$ for the smallest index~$j$), and, for the selected $d$, $C$ is the smallest
clause according to the ordering $\prec$.
%
% Otherwise, there is a clause
% $t: C$ that is minimal according to $\prec$ such that $\Imc_i\not\models t: C$. We
% first select a domain element $d\in\Delta^{\Imc_i}$ such that $d\not\in C^\Imc$. If
% $t=a$, this domain element is $a^\Imc$. Otherwise, we pick the
% \emph{oldest} domain element~$d$ such that $d\not\in C^\Imc$.
We then distinguish the following cases
based on the maximal literal $L$ in $C$ according to~$\prec_l$:
\begin{itemize}
 \item If $L=A\in\NC$, then $\Imc_{i+1}$ is obtained from $\Imc_i$ by adding
$d$ to $A^{\Imc_i}$,
 \item If $L=\exists r.D$, then $\Imc_{i+1}$ is obtained from $\Imc_i$ by
adding a new domain element $e$ and adding it to $D^{\Imc_i}$, as well as adding 
$\tup{d,e}$ to $r^{\Imc_i}$.
 \item If $L=\forall r.D$, then $\Imc_{i+1}$ is obtained from $\Imc_i$
by
adding each $r$-successor of $d$ to $D^{\Imc_i}$.
\end{itemize}
We observe that those steps ensure that $d\in C^{\Imc_i}$. 
We argue later that the case where $L=\neg A$ is not possible, so that those
cases are exhaustive. But before that, we would like to show the following
\emph{monotonicity property} of our construction: if $d$ is the domain element
selected for creating $\Imc_i$, and
for any clause $C'\in\Nbf^*$ s.t. $C'\prec C$, we have $d\in (C')^{\Imc_i}$,
then we also have $d\in (C')^{\Imc_j}$ for all $j>i$. This is clearly the case
if $d$ satisfies a literal of the form~$B$ or $\exists r.B$ in $C'$, since our
construction does not remove elements from the interpretations of concepts and
roles. If $d$ satisfies a literal of the form $\forall r.D$ in $C'$, we observe
that our ordering ensures that in no clause larger than $C'$, a literal of the
form $\exists r.D'$ can be maximal, so that no further successors can be added
to $d$.
Finally, if $d$ satisfies in $C'$ a literal of the form $\neg A$, we
first observe that our ordering ensures that $A$ cannot be maximal in a clause
larger than $C'$, since $A\prec\neg A$. Moreover, since in each step, we select
the \emph{oldest} domain element that still has unsatisfied clauses in
$\Nbf^*$, it is not possible that a predecessor of $d$ gets selected for some
$j>i$, since any predecessor would be older. Consequently, $d$ cannot be added
to $A$ due to a literal of the form $\forall r.D$ for any subsequent
interpretation.
% \patrick{Not ideal to argue like this: the fact that domain elements are
% processed one after the other only works because of the monotonicity
% property. But the monotonicity property is what we are trying to prove here.
% Better restructure the argument a bit.}

We now show that the case where $L=\neg A$ is not possible: since $d\not\in C^{\Imc_i}$,
$L=\neg A$ would mean that $d\in A^{\Imc_i}$. Assume such a case is possible, and let $i$ be the 
smallest index for which this happens. $d$ must have been added to $A^{\Imc_i}$ in an earlier iteration
for one of the following reasons.
  \begin{enumerate}
      \item There is some $j<i$ such that $d\not\in C_1^{\Imc_j}$, where $C_1\in\Nbf^*$, and
            $A$ is the maximal literal in $C_1$. Rule~\textbf{A1} is applicable
            on
            $C_1$
            and $C$, yielding a clause $C_2$ that is also in $\Nbf^*$. This clause is
            smaller than both $C$ and $C_1$, since $A$/$\neg A$ is
            maximal in
            these clauses and does not occur in $C_2$. Consequently, $C_2$ would
            have been
            processed before $\Imc_j$ by this procedure, which would have
            ensured that for some $k<j$, $d\in C_2^{\Imc_k}$. By the monotonicity property of our 
            construction, this means that $d\in C_2^{\Imc_j}$. 
            Now $C_2$ is composed exactly of the literals in $C/C_1$ except $\lnot A/A$, but $d\not\in C_1^{\Imc_j}$,
            which means that
            $d\in C^{\Imc_j}$. Due to the monotonicity property of our
            construction, then also $d\in C^{\Imc_i}$. But this contradicts
            that $d$ and $C$ are selected to obtain $\Imc_{i+1}$.
      \item $d$ is an $r$-successor of another domain element and was
            added due to some $j<i$ and selected clause $C_1$ not
            satisfied in $\Imc_j$ in which the maximal
            literal is $\exists r.A$. In this case, we have
            $d\neq d_0$. %, and thus $t=x$.
            The maximal literal in $C$ is then of the form $\neg A$, with $A$ under
            an existential role restriction in $\Nbf$.
            By our ordering, this means that all other literals in $C$ must be
            negated (Condition~\ref{cond:definers}) and their concept names
            occur in $\Nbf$ under existential role restrictions
            (Condition~\ref{cond:exists}).
            We can indeed conclude from this that
            $\neg A$ is the only literal in $C$. Otherwise, there would be
            another literal $\neg B$ in $C$, where $B$ does not occur in a
            value restriction.
            (Recall that by our assumptions, if $B$ occurs positively under an existential role restriction,
            it cannot occur positively in a different way). Since $d\not\in C^{\Imc_i}$, also $d\in
            B^{\Imc_i}$, but there is no step in the construction that would add
            an existing domain element to a concept occurring under an
            existential role restriction.
            
            We obtain that $\neg A\in\Nbf^*$.
            But then, the \textbf{r2}-rule applies on this clause and $C_1$
            for the
            case of $n=0$,
            resulting in a clause $C_2$ that is obtained from
            $C_1$ by removing its maximal literal. Thus, $C_2\prec C_1$,
            which means that $C_2$ must have been processed before $C_1$.
            This in turn means that $\Imc_j\models C_2$, and thus
            $\Imc_j\models C_1$, so $C_1$ could not have been
            selected in Step~$j$.

      \item $d$ is an $r$-successor of another domain element $e$ and was
            added to $A^{\Imc_i}$ due to some $j<i$ and clause $C_1$ in
            which the maximal literal is $\forall r.A$. As before, we can argue
            that
            all literals in $C$ are of the form $\neg D$, with $D$
            occurring under a role restriction,
            and $C$ contains at most one literal $\neg D_\exists$, where
            $D_\exists$ occurs under an existential role restriction. In
            particular, $d$ was created as an $r$-successor of $e$
            due to a clause in which $\exists r.D_\exists$ is the maximal
            literal, and for every $\neg D$ in $C$, $d$ was added to the
            interpretation of $D$ due to a clause in which the maximal literal
            is $\forall r.D$. We observe that one of the rules \textbf{r1} or
            \textbf{r2} is applicable on those clauses together with $C$,
            resulting in a clause that is smaller and consequently must have
            been processed before all the other clauses, making at least one of
            these clauses satisfied for $e$, and contradicting that this clause
            was used to add $d$ to the interpretation of some $D$ such that $\neg D$
            occurs in $C$.
  \end{enumerate}
  As a consequence, we obtain that, in each step, one clause is satisfied
for one domain element for which it was not satisfied before. In the
limit, we obtain that $\Imc$ satisfies all clauses in $\Nbf^*$, and thus is a model of $\Nbf$.\qed
\end{proof}
% \patrick{@Stefan: the problems you found in the proof were because I didn't specify the ordering properly (there was even a syntactic typo). Also, I realized I forgot to argue that the construction is monotonic, which does not directly follow since clauses may contain negated literals that could in theory become satisfied in subsequent steps. Now the proof should work.}

\subsection{Optimizations}

From the construction in the proof of \Cref{the:alc-calculus}, we see that some clauses in $\Nbf^*$ are not
relevant for the refutational completeness, and can thus be discarded:
\begin{enumerate}
 \item \emph{Tautologies}, that is, clauses containing both $A$ and $\neg A$ for some $A\in\NC$
 are always satisfied and will thus never trigger an adaptation of the current interpretation. In our implementation,
 tautologies are never added to the current set of clauses.
 \item \emph{Subsumed clauses}, \ie clauses $C_1$ such that, for some other clause $C_2\in\Nbf^*$, $C_2\neq C_1$
    and every literal in $C_2$ also occurs in $C_1$. By our ordering, $C_2\prec C_1$, so that our model
    construction will consider $C_2$ before it considers $C_1$. By the monotonicity property, satisfiying $C_2$
    furthermore ensures that $C_1$ remains satisfied, and is thus never considered by the model construction.
    In our implementation, we use both \emph{forward subsumption}, that is, newly derived clauses that are
    subsumed by previously derived clauses are not added to the current set of clauses, and \emph{backward
    subsumption}, that is, after adding a new clause, we remove from the current set of clauses all
    clauses that are subsumed by it.
\end{enumerate}

We furthermore observe that the arguments in the proof only consider the maximal literals in a clause according
to the literal ordering $\prec_l$. For this reason, our method remains refutationally complete if we only perform
inferences on the maximal literal.

\subsection{Generating \ALC Proofs}

Fix an \ALC ontology $\Omc$ and concepts $C$, $D$. To compute a proof for $\Omc\models C\sqsubseteq D$, we would
first introduce concept names $A_C$, $A_D$ for those concepts, and add the axioms $A_C\sqsubseteq C$, $D\sqsubseteq A_D$
(that is, we reduce to subsumption between concept names). As mentioned in the text, we use concept names
$A_\text{LHS}$, $A_\text{RHS}$ to track inferences for the proof. After normalizing, we add the clauses
$A_C\sqcup A_\text{LHS}$ and $\neg A_D\sqcup A_\text{RHS}$, which we add to the set of clauses. We make sure that
$A_\text{LHS}$ and $A_\text{RHS}$ are minimal in the ordering used, so that the clause $A_\text{LHS}\sqcup A_\text{RHS}$
or a subclause is derived if $\Omc\models C\sqsubseteq D$ (a subclause in the cases where the ontology is inconsistent
(empty clause), $C$ is unsatisfiable, or $\Omc\models\top\sqsubseteq D$). By tracking all inferences, we obtain a proof
for that entailment from $\Nbf$,
in which nodes are labeled with clauses. This proof still has to be transformed into a more readable DL proof for
$\Omc\models C\sqsubseteq D$. For this, we
\begin{itemize}
 \item regard clauses $C$ as GCIs $\top\sqsubseteq C$,
 \item replace $A_{LHS}$ everywhere by $\neg C$ and $A_{RHS}$ by everywhere by $D$, where $C$ and $D$ are the concepts
    in the subsumption $C\sqsubseteq D$ to be proved,
 \item replace $A_C$ by $C$ and $A_D$ by $D$, where $C$ and $D$ are the concepts of the subsumption to be proved,
 \item replace any concept names $A$ under role restrictions $\quant r.A$, which were used during normalization to
    flatten expressions $\quant r.C$, again by $C$,
 \item add additional inferences to link the non-tautologigal leafs of the proof to axioms from the ontology (from
 which they were obtained during normalization),
 \item apply the following transformations exhaustively on each clause to make them more human-readable, where
  we treat $\top$ as empty conjunction and $\bot$ as empty disjunction:
  \begin{itemize}
    \item $C\sqsubseteq\neg A\sqcup D \qquad \Longrightarrow \qquad C\sqcap A\sqsubseteq D$
    \item $C\sqsubseteq\forall r.\bot\sqcup D\qquad \Longrightarrow \qquad C\sqcap\exists r.\top\sqsubseteq D$
    \item $C\sqsubseteq\exists r.\neg D\sqcup E\qquad \Longrightarrow \qquad C\sqcap\forall r.D\sqsubseteq E$
    \item $C\sqsubseteq\forall r.\neg D\sqcup E\qquad \Longrightarrow \qquad C\sqcap\exists r.D\sqsubseteq E$
  \end{itemize}
\end{itemize}
The last step is not strictly necessary, but the final goal of the proof is to explain
the inference to the user, for which we want to minimize the number of negations used. As a result of the transformation, the clauses $A_C\sqcup A_\text{LHS}$ and
$\neg A_D\sqcup A_\text{RHS}$ that were added to the initial clause set get transformed into tautologies
$C\sqsubseteq C$ and $D\sqsubseteq D$.

\subsection{Integrating the Concrete Domain}

The reasoning algorithm keeps a set $\textbf{D}$ of currently relevant constraints. To compute all relevant
implications $\alpha_1\wedge\ldots\wedge\alpha_n\rightarrow\beta$, where $\alpha_1,\ldots,\alpha_n\in\textbf{D}$,
we use the proof procedure for the concrete domain to compute the set of possible inferences starting from
$\textbf{D}$. For each derived constraint $\beta$ for which $A_\beta$ occurs negatively in the set of clauses, we
follow those inferences back to the constraints in $\textbf{D}$ that were used to derive it, creating all such possible
sets. This way, we ensure that all implications $\alpha_1\wedge\ldots\wedge\alpha_n\rightarrow\beta$,
where $\{\alpha_1,\ldots,\alpha_n\}\subseteq\textbf{D}$ is subset-minimal, are discovered, so that we can add the corresponding
clauses to the current set of clauses. The procedure might also create some implications where the
set of constraints on the left-hand side is not subset-minimal. However, those clauses are immediately removed due
to forward-subsumption deletion.

\ThmALCDReasoning*
\begin{proof}
  The complexity of the method does not change, since we only add exponentially
many new clauses. Soundness follows from the fact that the added clauses all
correspond to valid entailments for the original ontology $\Omc$ if we replace
the fresh concept names again by the corresponding concrete constraints. For
completeness, we use the construction as in the proof of
\Cref{the:alc-calculus} to construct a model $\Imc$ for $\Nbf^*$, the final set
of clauses, in case $\bot\not\in\Nbf^*$.
$\Imc$ is a model for $\Omc^{-\Dmc}$, but we still have to modify it into a
model of $\Omc$. We argue that we can do that by ensuring that for every
$d\in\Delta^\Imc$:
\begin{enumerate}
 \item if $d\in A_D^\Imc$, then $d\in D^\Imc$,
 \item if $d\not\in A_D^\Imc$, and there is no clause $t: C\in\Nbf^*$ s.t. $t$
can be interpreted by $d$ (i.e., either $t=x$, or $t=a$ and $d=a^\Imc$), the
maximal literal of $C$ is $\neg A_D$, and $d\not\in L^\Imc$ for any other literal in
$C$, then $d\not\in D^\Imc$.
\end{enumerate}
If we can modify $\Imc$ like that, then we obtain a model of $\Omc$.
%Furthermore, $\Imc'\not\models A\sqsubseteq B$ since $a^{\Imc'}\in A^{\Imc'}$
%and $a^{\Imc'}\not\in B^{\Imc'}$, so that $\Imc'$ would witness that
%$\Omc\not\models A\sqsubseteq B$. It remains to show that we can really
%modify $\Imc$ in such a way.

For a domain element $d\in\Delta^\Imc$, we collect the
sets $\Dbf$ and $\overline{\Dbf}$, where $\Dbf=$ $\{D\in\Cmc(\Omc)\mid d\in
A_D^\Imc\}$ and $\overline{\Dbf}$ contains all the $D\in\Cmc(\Omc)$ for which
$A_D$ satisfies the second condition above.

We first observe that
$\Dmc\not\models\bigwedge\Dbf\rightarrow\bot$, since the concept names
corresponding to the constraints in $\Dbf$ must occur as maximal literals in
some clauses in $\Nbf^*$, and we would thus have added a clause $x:\bigsqcup_{D\in\Dbf'}\neg A_D$ for some subset $\Dbf'\subseteq\Dbf$ if $\Dbf$ was inconsistent. This clause would not be
satisfied by $\Imc$, contradicting that $\Imc\models\Nbf^*$. We can thus extend
$\Imc$ in such a way that, for each $D\in\Dbf$, $d\in D^{\Imc'}$. It remains to
show that we can also ensure that for every $D\in\overline{\Dbf}$, $d\not\in
D^{\Imc'}$. We first observe that for no $D\in\overline{\Dbf}$,
$\Dmc\models\bigwedge\Dbf\rightarrow D$. This is because $\overline{\Dbf}$ was
carefully chosen to ensure that if $\Dmc\models\bigwedge\Dbf\rightarrow D$,
there would be some subset $\Dbf'\subseteq\Dbf$ for which we would
have added the clause $x:\bigsqcup_{D'\in\Dbf'}\neg A_{D'}\sqcup A_D$.
Since $d$ does not satisfy the concept, this would contradict that $\Imc$ is a
model of $\Nbf^*$. Because $\Dmc$ is convex, it follows furthermore that
$\Dmc\not\models\bigwedge\Dbf\rightarrow\bigvee\overline\Dbf$.
Consequently, we can find an assignment of the
concrete
features that ensures that $d\in D^{\Imc'}$ for all $D\in\Dbf$ and $d\not\in
D^{\Imc'}$ for all $D\in\overline{\Dbf}$. It follows that we can extend $\Imc$
to a model of $\Omc$ as required.\qed
\end{proof}

To produce a combined proof in \ALCD, we again follow the approach described 
in Section~\ref{sec:combining-proofs}, but for now we always use the left-hand 
side $A_{\alpha_1}\sqcap\dots\sqcap A_{\alpha_n}$ as the context for the CD 
proof of $\alpha_1\land\dots\land\alpha_n\to\beta$, since it is not trivial to 
find other meaningful DL contexts for proofs generated from the calculus in 
Fig.~\ref{fig:alc-calculus}.

% \stefan{TODO: describe a bit more about how to obtain proofs, since it is not quite the same as in \ELbotD}

% \subsection{Optimizations}

% \begin{itemize}
%  \item clauses to avoid
%  \item tautologies
%  \item set-of-support
%  \item compute justifications for the concrete domain
% \end{itemize}

\subsection{The Complexity of Finding Good Proofs}

To investigate the complexity of our approaches and prove Theorem~\ref{thm:complexitygoodproofs}, we use the formal framework from~\cite{DBLP:conf/lpar/AlrabbaaBBKK20,DBLP:conf/cade/AlrabbaaBBKK21}, which we shortly introduce in the following.
A \emph{derivation structure} for an entailment $\Tmc\models\eta$ in a logic~\Lmc is a directed, labeled hypergraph~$(V,E,\ell)$ where
\begin{enumerate}
  \item vertices are labeled with $\Lmc$-sentences, %edge labels indicate inference types,
  \item every leaf is labeled by an axiom from $\Tmc$, and
  \item every hyperedge $(S,d)\in E$ is an inference satisfying $\{\ell(v)\mid
 v\in S\}\models\ell(d)$.
\end{enumerate}
Here, a \emph{leaf} is a node $v\in V$ without an incoming hyperedge $(S,v)\in E$, and a \emph{sink}~$v\in V$ has no outgoing hyperedges $(S,d)\in E$ with $v\in S$.
A \emph{proof} for $\Tmc\models\eta$ is such a derivation structure that, additionally,
\begin{enumerate}[resume]
  \item is tree-shaped, \ie has no cycles in the relation $\{(s,d)\mid (S,d)\in E,\ s\in S\}$,
  \item has a unique sink labeled by the final conclusion~$\eta$, and
  \item has no two hyperedges $(S,v),(S',v')\in E$ with $\ell(v)=\ell(v')$.
\end{enumerate}
 
As in \cite{DBLP:conf/lpar/AlrabbaaBBKK20,DBLP:conf/cade/AlrabbaaBBKK21}, we consider a so-called \emph{deriver}~\R, which produces derivation structures $\R(\Tmc,\eta)$ that contain all inference steps relevant for a proof of $\Tmc\models\eta$ (but they can encompass many possible ways of deriving~$\eta$).
We are interested in finding a proof that can be homomorphically mapped into $\R(\Tmc,\eta)$, and whose \emph{size} (number of vertices) is below a given threshold (a \enquote{small proof}).
Reasoners for \ELbot, such as \Elk, produce derivation structures of polynomial size, and the problem of finding small proofs in such structures is \NP-complete \cite{DBLP:conf/lpar/AlrabbaaBBKK20}.
For derivation structures that are of exponential size in general, such as for \ALC, this problem is \NExpTime-complete \cite{DBLP:conf/lpar/AlrabbaaBBKK20}.
If we replace \emph{size} by \emph{recursive} measures such as \emph{depth} (the length of the longest path from the sink to a leaf), the complexity drops to \PTime and \ExpTime, respectively \cite{DBLP:conf/cade/AlrabbaaBBKK21}.

To determine the complexity of finding small proofs in \ELbotD and \ALCD, we additionally need to consider the size of the derivation structures for \DQlin and \DQgr, which are then integrated into the pure DL proofs.
For \DQgr, Theorem~\ref{thm:cd2-reasoning} shows that the derivation structures constructed by Algorithm~\ref{alg:cd2-reasoning} are always of polynomial size, which does not change the complexity of finding proofs compared to the case of DLs without concrete domains.
For \DQlin, however, derivation structures can be of exponential size: Although Gaussian elimination produces at most quadratically many inference steps in the number of variables that occur in the constraints, there are exponentially many possible orders in which the variables could be eliminated and different choices of constraints to use for eliminating a variable, each of which yields a different proof.
To alleviate this problem in our implementation, we normalize all equations after each elimination step, by reducing the coefficient of the leading variable (according to a fixed variable order) to~$1$, and adjusting the other coefficients accordingly.
For example, the elimination step
\[
  \AXC{$4x-6y=1$}
  \AXC{$2x+3y=5$}
  \RL{\scriptsize$[1,-2]$}
  \BIC{$-12y=-9$}
  \DP
\]
from the previous example would become
\[
  \AXC{$4x-6y=1$}
  \AXC{$2x+3y=5$}
  \RL{\scriptsize$[-\frac{1}{12},\frac{1}{6}]$}
  \BIC{$y=\frac{3}{4}$}
  \DP
\]
Nevertheless, the overall size of the derivation structure stays exponential in the worst case.

\complexitygoodproofs*
\begin{proof}
  In the cases involving \DQgr or \ALCD, we obtain polynomial (exponential) structures for \ELbotD (\ALCD) entailment problems by transforming the concrete domain derivation structures and integrating them into the classical DL derivation structures as described in Lemma~\ref{lem:transformation}.
  In \ELbotD, our method considers only polynomially many \Dmc-implications, while for \ALCD we obtain exponentially many \Dmc-derivation structures (of polynomial or exponential size), which does not affect the exponential size of the derivation structures for \ALC.
  Hence, we can apply the classical algorithms for finding proofs of a given maximal size or depth in the combined structures~\cite{DBLP:conf/lpar/AlrabbaaBBKK20,DBLP:conf/cade/AlrabbaaBBKK21}.
  
  For \ELbotD[\DQlin], the idea is to guess a substructure of the combined derivation structure in polynomial time, and then verify that it is indeed a proof of the required size.
  Since the \ELbot-parts of the derivation structure are of polynomial size, we can guess those parts in \NP.
  For all guessed axioms $\bigsqcap_{D\in\Dbf_C}A_D\sqsubseteq A_E$, we then need to guess a corresponding proof of $\bigwedge\Dbf_C\to E$ in \DQlin.
  However, we know that such a proof needs at most polynomially many variable elimination steps (at most one for each variable in each involved constraint), which correspond to inference steps.
  Hence, we can guess in polynomial time a variable elimination order and, for each constraint~$\alpha$ and variable~$x$, a constraint that is used to eliminate~$x$ from~$\alpha$.
  \qed
\end{proof}

%!TeX root=../main.tex
\section{Implementation and Experiments}
We implemented the algorithms described above and evaluated their performance and the produced proofs on a series of benchmarks.
The implementation uses the Java-based OWL API 4 to interact with DL ontologies, but uses new data structures for representing concrete domains.
Although there is a proposal for extending OWL with concrete domain predicates of arities larger than~$1$,\footnote{\url{https://www.w3.org/2007/OWL/wiki/Data_Range_Extension:_Linear_Equations}} this is not part of the OWL 2 standard.\footnote{\url{https://www.w3.org/TR/owl2-overview/}}
To extract proofs from the collections of inference steps produced by our concrete domain reasoning algorithms, we used a Dijkstra-like algorithm that minimizes the size of the produced proof~\cite{DBLP:conf/cade/AlrabbaaBBKK21}.
For efficiency reasons, for proofs in \DQlin, we fix a variable order for the Gaussian elimination steps, instead of considering all possible orders in which variables could be eliminated.

% \subsection{Proof Shapes}\label{appendix:proofpictures}

Returning to Example~\ref{ex:medical}, we can split it into two tasks to demonstrate proofs in both concrete domains, where we have added information
on the status of our current patient using the GCIs $\ex{CurrentPatient} 
\sqsubseteq [\ex{age} = 42]$, $\ex{CurrentPatient} \sqsubseteq [\ex{hr} = 
173]$, and $\ex{CurrentPatient} \sqsubseteq [\ex{pp} = 65]$. Resulting proofs 
of  $\ex{ICUpatient} \sqsubseteq \ex{NeedAttention}$ are shown in 
Fig.~\ref{app:fig:patientcd1} and Fig.~\ref{app:fig:patientcd2}. 
\begin{figure}
\centering
\includegraphics[width=\textwidth]{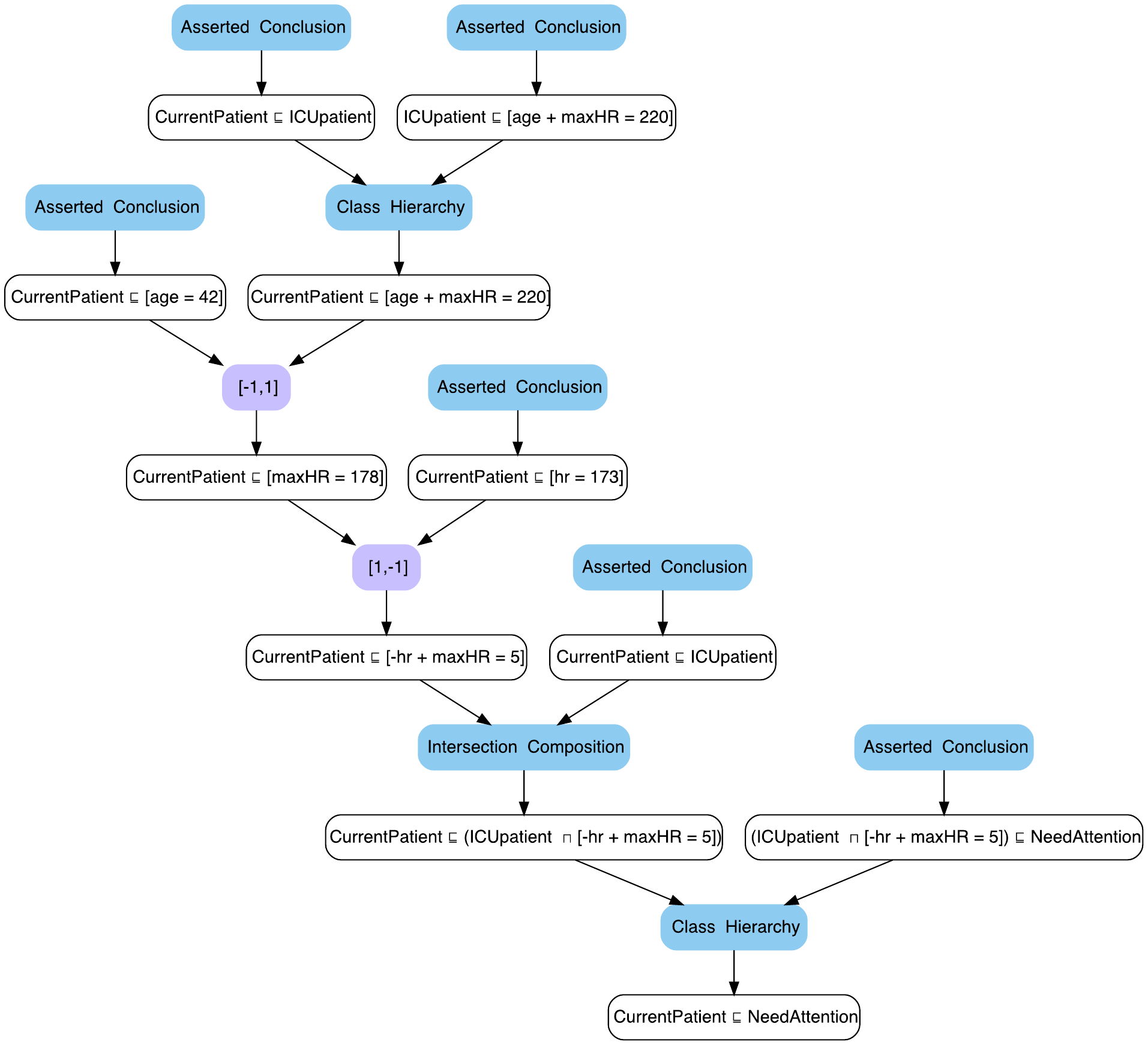}
\caption{Showing $\ex{CurrentPatient} \sqsubseteq \ex{NeedAttention}$ using \DQlin} \label{app:fig:patientcd1}
\end{figure}
\begin{figure}
\centering
\includegraphics[width=0.6\textwidth]{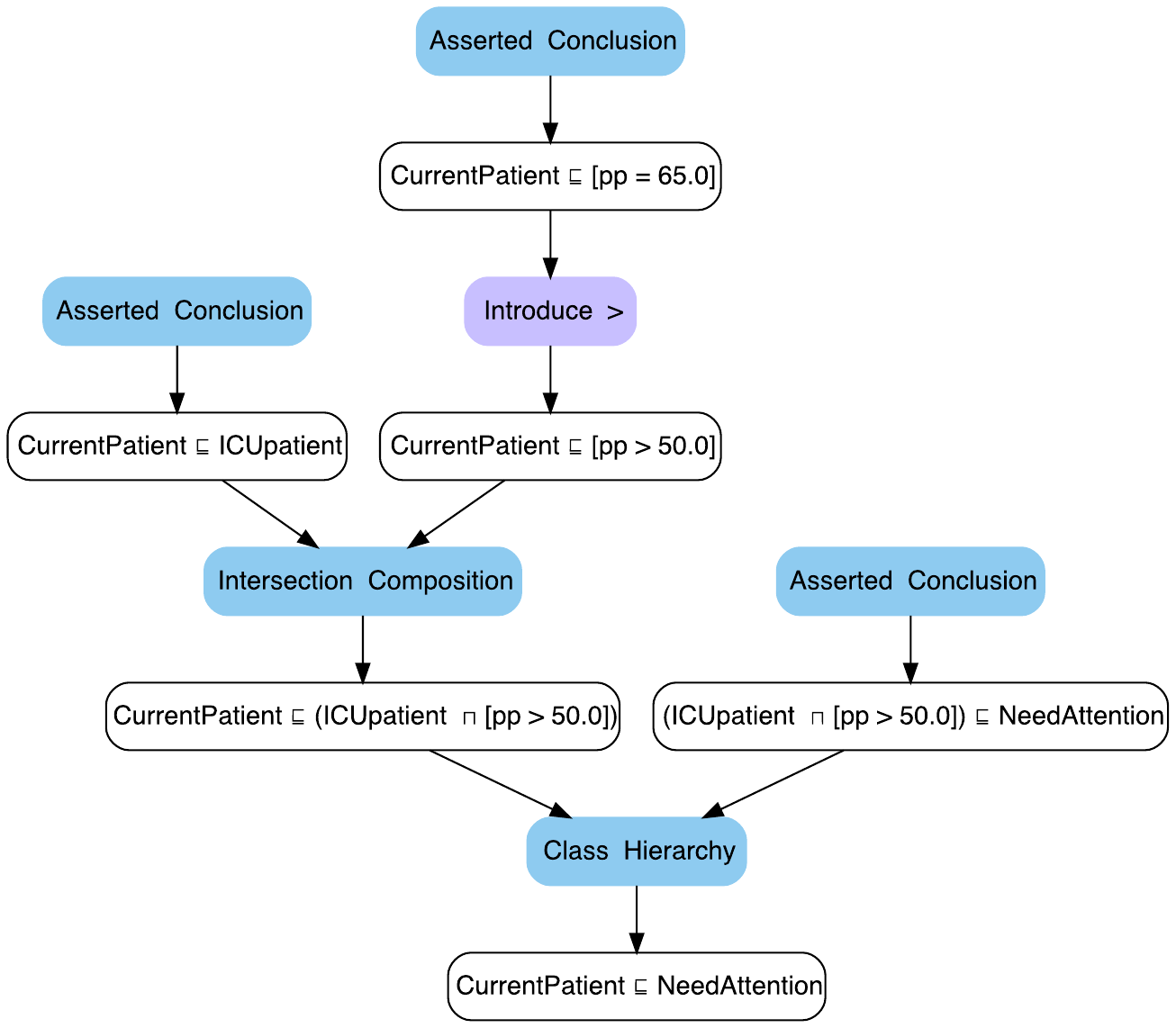}
\caption{Showing $\ex{CurrentPatient} \sqsubseteq \ex{NeedAttention}$ using \DQgr}\label{app:fig:patientcd2}
\end{figure}

\subsection{Benchmarks}

In the following, we describe several ontologies that we developed to evaluate 
our implementation.
Unfortunately, existing reasoning tasks for DLs with concrete domains, \eg from 
Racer,\footnote{\url{https://github.com/ha-mo-we/Racer}} are not expressive 
enough to test our algorithms; the CD values are used only as constants and do 
not influence the reasoning.
Some of our benchmarks are scalable in the sense that they are based on similar ontologies, but one can increase their size, \eg by increasing the number or size of axioms or constraints.
All benchmarks are formulated in \ELbotD.
% \stefan{Adapt if this changes.}

\paragraph{Simple benchmarks.} For \DQlin, there are two basic demonstration examples, \emph{Drones} and \emph{Coffee}. The main task in \emph{Drones} is to derive the fraction of impaired sensors and propellers of a drone.

In \emph{Coffee}, for different types of coffee such as cappuccino, ristretto, macchiato, \etc, we define
% \footnote{the ontology and the constraints are based on \url{https://www.teacoffeecup.com/recipe/different-types-of-coffee-explained-espresso-drink-recipes/}}
the proportions of components such as espresso, steamed milk, foam, \etc Consequently, we can identify a coffee drink based on the amounts of its components given in some unit like ml or oz.

\paragraph{Scalable benchmarks.} For testing the system behavior on inputs of 
increasing size, we provide four benchmarks, two for \DQlin and two for \DQgr. 

In \emph{Diet}$(n)$, given a person's daily consumption as a list of $n$ products and their calories from fat, protein, and carbs, we check constraints about the consumption, such as \enquote{full fat}, \enquote{full protein}, \enquote{full carbs}, \enquote{well-balanced} (55\% carbs, 20\% protein, 25\% fat), \enquote{lower carb} (45\% carbs, 25\% protein, 30\% fat), and \enquote{lower carb and fat} (45\% carbs, 30\% protein and 25\% fat), expressed in \DQlin. The parameter $n$ describes the size of the linear constraints.

To scale both the DL and CD parts, we created the benchmark \emph{Artificial}$(n)$ over \DQlin. With increased~$n$, the number of intermediate concepts in the concept hierarchy and in proofs also increases, \ie $A\sqsubseteq C_0 \sqsubseteq \dots \sqsubseteq C_{n-1} \sqsubseteq B $. Moreover, to show each step, the concrete domain reasoner has to derive a linear constraint from $n$ given constraints. Thus, the size of a proof for $A \sqsubseteq B$ grows quadratically in~$n$.
%?

For \DQgr, there are \emph{D-Sbj}$(n)$ and \emph{D-Obj}$(n)$. In \emph{D-Sbj}$(n)$, there is one main actor, which is a drone. The parameter $n$ quantifies the number of other objects in the world. The reasoner needs to show that all objects are at a \enquote{safe} distance from the drone.

The benchmark \emph{D-Obj}$(n)$ is somehow orthogonal to \emph{D-Sbj}$(n)$. The world contains $n$ drones and exactly $3$ other objects, \eg humans or trees. Now the distances between all drones and objects are taken into consideration. Similarly to the subjective version, the reasoner needs to find out whether all objects are at \enquote{safe} distances.

\paragraph{Experiments.}

Our findings for the \ELbotD algorithms are summarized in Tables~\ref{tab:results} and~\ref{tab:description}.
The first two benchmarks consist of a single reasoning problem each, the remaining four represent series of small to medium-sized ontologies, for which we report ranges in each column.
For the scalable benchmarks, we consider only the instances which terminated and, for each size, computed average times over three randomly generated instances of that size.
The underlined benchmarks are in \DQgr, the rest in \DQlin.
Columns 2--4 show the number of axioms, number of constraints, and average number of variables per constraint in the ontology, respectively.
Column 5 aggregates the number of occurrences of any name (concept, role or feature) in the ontology.
Columns 6--7 list the time (in ms) required for classification and proof generation, respectively,
% %
% \enquote{\%DL} denotes the fraction of the reasoning time that was used by the calls to \Elk.
% %
% \enquote{\%Incr} denotes the fraction of \Elk reasoning time that would have been required for a single call to \Elk on the final saturated ontology~$\Omc'$.
% %
% \enquote{Expl} denotes the time it took to produce the proof (in ms). Finally,
\enquote{$\mathrm{\Sigma}$} denotes the total number of instances, and
\enquote{\#Fin} denotes the number of instances for which proof generation finished before a timeout of 3\,min.
\begin{table}[tb]
  \centering
  \caption{Results of the experiments with the \ELbotD algorithms.
  }
  \label{tab:results}
  % \begin{tabular}{lcccccccr}
  %   \toprule
  %   Name & \#Ax & \#Constr &  \#Vars & Time & \%DL($\mathit{sd}$) & \%Incr($\mathit{sd}$)  & Expl & \#Fin/$\mathrm{\Sigma}$ \\
  %   \midrule
  %   Coffee & 49 & 21 & 2.3 & 867 & 55 & 18 & 46 & 1/1 \\
  %   Drones & 93 & 11 & 2 & 272 & 86 & 61 & 41 & 1/1  \\
  %   \midrule
  %   Diet  & 26--194 & 15--99 & 2--40 & 67--318 & 24(7) & 61(7)& 73--880 & 24/24 \\
  %   %164(67,318) & 24(15,38)& 61(53,70)&  445(73,880) &100\%  \\
  %   Artificial & 9--24 & 3--13 & 3--5 & 20--969 & 71(9) & 56(33) & 129--258 & 10/18\\
  %   %372(20,969) & 71(62,81) & 56(23,88) &175(129,258)& 56\%\\
  %   \underline {D-Sbj} & 55--279 & 20--160 & 1.5 & 177--850 & 74(2)&  71(17)& 83--7082 & 24/24\\
  %   %494(177,850) & 74(70,77)&  71(37,89)& 3586(83,7082)&100\%\\
  %   \underline{D-Obj} & 73--530 & 17--225  & 1.5 &  117--2990  & 80(7)& 80(13)& 166--13004 & 27/33\\
  %   %986(117,2990) & 80(69,88)& 80(55,102)& 4846(166,13004)& 81.25\%\\
  %   \bottomrule
  % \end{tabular}
  \begin{tabular}{lccccccc}
    \toprule
    Name & Axioms & Constraints & Variables & Problem Size & Time & Proof Time & \#Fin/$\mathrm{\Sigma}$ \\
    \midrule
    \emph{Coffee} & 49 & 21 & 2.3 & 146 & 329 & 109 & 1/1 \\
    \emph{Drones} & 93 & 11 & 2 & 254 & 195 & 100 & 1/1  \\
    \midrule
    \emph{Diet}  & 26--194 & 15--99 & 2--40 & 81--865 & 61--321 & 
    45--2659 & 8/8 \\
    \emph{Artificial} & 9--24 & 2--13 & 3--5 & 25--144& 17--347 & 18--  & 4/6 
    \\
    \underline{\emph{D-Sbj}} & 55--204 & 20--85 & 1.5 & 191--1101 & 
    220--686 & 166--482 & 8/8 \\
    \underline{\emph{D-Obj}} & 122--427 & 24--76 & 1.5 & 427--2129 & 
    315--1437 & 417--2099 & 8/8 \\
    %986(117,2990) & 80(69,88)& 80(55,102)& 4846(166,13004)& 81.25\%\\
    \bottomrule
  \end{tabular}
\end{table}
  % \alisa{\%DL = OWL Classification Duration (ms) / Classification Time}
  %   \alisa{\%Incr = ELK Classification Duration of Saturated Ontology (ms)/OWL Classification Duration (ms)}
% \alisa{ranges are growing wrt to instance size}
% \stefan{exponential blowup for proofs in \DQlin expected since we consider all possible ways of choosing constraints to eliminate the next variable (according to a fixed order); finding only one (random) proof is much faster; in general all proofs have $\frac{n^2}{2}$ steps anyway since we need to eliminate all variables from all constraints (except the ones we used for eliminating)}

\begin{table}
  \centering
  \caption{Supplemental experiment data for the \ELbotD algorithms (\cf Table~\ref{tab:results}).
  \enquote{\%DL} denotes the fraction of the reasoning time that was used by the calls to \Elk.
  \enquote{\%Incr} denotes the fraction of \Elk reasoning time that would have been required for a single call to \Elk on the final saturated ontology~$\Omc'$.
  %
  % The first column contains the name of each benchmark, the second column the size variation, then the used CD, columns 4--5 (the ranges in) the numbers of axioms and constraints, respectively.
  %
  % The next column shows the ranges in the average number of variables per constraint in the ontology.
  %
  The last two columns describe the computed proofs in terms of their tree size 
  and \enquote{\%CD}, the average fraction of the steps of the computed proof 
  that are due to concrete domain inference steps.
  For the scalable benchmarks, we report either ranges or averages with standard deviations (\textit{sd}) in each column.}
  \label{tab:description}
  \centering
  \begin{tabular}{lcccccc}
    \toprule
    Name & \%DL(\textit{sd}) & \%Incr(\textit{sd}) & Proof Size & \%CD(\textit{sd}) \\
    \midrule
    \emph{Coffee} & 41 & 20 & 18 & 11 \\
    \emph{Drones} & 83 & 51 & 7 & 11 \\
    \midrule
    \emph{Diet} & 23(7) & 58(6) & 22--166 & 37(10) \\
    \emph{Artificial} & 75(15) & 78(16) & 10--51 & 18(5) \\
    \emph{D-Sbj} & 52(2) & 76(8) & 45--584 & 6(0) \\
    \emph{D-Obj} & 63(2) & 86(6) & 184--1232 & 8(0) \\
    \bottomrule
  \end{tabular}
\end{table}

%\begin{figure}
%\includegraphics[width=\textwidth]{pictures/OWLClassificationvsELKClassification.pdf}
%\caption{Comparison of OWL classification duration (horizontal axis) to ELK classification duration of saturated ontology (vertical axis) in ms. Data is averaged among three instances of same size to decrease the influence of random numbers in CD constraints.}
%\end{figure}

}
\fi
\end{document}